\documentclass[12pt]{article}
\usepackage{amsmath}
\usepackage{graphicx,psfrag,epsf}
\usepackage{enumerate}
\usepackage{natbib}
\usepackage{url} 

\newcommand{\blind}{0}

\usepackage{multirow, booktabs}
\usepackage{amsthm,amssymb,bm}
\graphicspath{{./figure/}}

\newtheorem{theorem}{Theorem}[section]

\newtheorem{assumption}{Assumption}[section]
\newtheorem{lemma}{Lemma}[section]
\newtheorem{proposition}{Proposition}[section]

\newcommand{\tr}{\text{Tr}}
\newcommand{\cov}{\text{Cov}}
\newcommand{\diag}{\text{Diag}}
\newcommand{\card}{\text{Card}}

\begin{document}

	\def\spacingset#1{\renewcommand{\baselinestretch}%
		{#1}\small\normalsize} \spacingset{1}

	
	\if0\blind
	{
		\title{\bf Adaptive Testing for Alphas in Conditional Factor Models with High Dimensional Assets}
		\author{Huifang Ma\\
			Nankai University
			\\~\\
			Long Feng\hspace{.2cm} \\
			Nankai University
		    \\~\\
	        Zhaojun Wang\hspace{.2cm} \\
	        Nankai University}
		\maketitle
	} \fi
	
	\if1\blind
	{
		\bigskip
		\bigskip
		\bigskip
		\begin{center}
			{\LARGE\bf  Adaptive Testing for Alphas in Conditional Factor Models with High Dimensional Assets}
		\end{center}
		\medskip
	} \fi
	
	\bigskip
	\begin{abstract}
This paper focuses on testing for the presence of alpha in time-varying factor pricing models, specifically when the number of securities $N$ is larger than the time dimension of the return series $T$. We introduce a maximum-type test that performs well in scenarios where the alternative hypothesis is sparse. We establish the limit null distribution of the proposed maximum-type test statistic and demonstrate its asymptotic independence from the sum-type test statistics proposed by \citet{Ma2018TestingAI}. Additionally, we propose an adaptive test by combining the maximum-type test and sum-type test, and we show its advantages under various alternative hypotheses through simulation studies and two real data applications.
\end{abstract}
	
	\noindent%
	{\it Keywords:}  Alpha tests, Conditional factor model, High dimensionality, Sparse alternatives, Spline estimator
	\vfill
	
\newpage
\spacingset{1.8} 
\section{Introduction}
\label{sec:intro}
    Tests of alpha have attracted much attention in the econometrics literature. In the traditional linear pricing factor models, e.g. CAPM \citep{Sharpe1964CAPITALAP,Lintner1965THEVO} and the Fama-French three-factor model (abbreviated as FF model hereafter)  \citep{Fama1993CommonRF,Fama2014AFA}, \citet{Gibbons1989ATO} proposed an exact multivariate F-test under the joint normality assumption. Several extended methods include \citet{Mackinlay1991UsingGM}, \citet{Zhou1993AssetpricingTU} and \citet{Beaulieu2007}, etc. However, their application has been confined to a relatively small number of portfolios. With the development of modern financial markets, thousands of securities are traded. So the assumption of fixed dimension is not appropriate. Consequently, recent efforts have focused on developing tests that allow the number of securities $N$ are larger than the time periods $T$. For instance, \citet{Pesaran2012TestingCW,Pesaran2017TestingFA} propose a sum-type test statistic by replacing the sample covariance matrix with identity matrix in traditional F-test, which performs well under the dense alternatives. For sparse alternatives, \citet{Feng2022AsymptoticIO} proposed a max-type test statistics. \citet{Yu2023PE} employed the thresholding covariance estimator of \citet{Fan2011LargeCE} and proposed a novel Wald-type test. They also proved the asymptotic independence between the new Wald-type test and the maximum-type test and proposed a new Fisher combination test which performs very well under neither dense nor sparse alternatives.
	
    Although these tests have addressed the limitation of $T>N$, they still require the factor loadings be constant over time. This assumption can be quite restrictive in empirical finance. Much empirical evidence indicates that the factor loadings in CAPM and the FF model vary substantially over time even at the portfolio level \citep{Conover2007TheCC,Chen2003CAPMOT}. As a result, the aforementioned tests can lead to inaccurate conclusions when the factor loadings are time-varying. To remove the limiting of time-invariant factor loadings, \citet{Li2011TestingCF} and \citet{ANG2012132} proposed nonparametric Wald-type tests under the case that $N$ is fixed and $T$ goes to infinity, and \citet{Gagliardini2016TimeVaryingRP} developed an econometric methodology under high dimensional case. Furthermore, \citet{Ma2018TestingAI} proposed a sum-type statistic based on the residuals obtained from the null model, which is asymptotically normal-distributed.

    The above methods are all sum-of-squares types, generally having good power performance against dense alternatives. However, inefficient market pricing is more likely to occur in exceptional assets rather than systematic mispricing of the entire market, i.e. $\bm{\alpha}$ has few nonzero elements with a large $||\bm{\alpha}||_\infty$. In this case, it is problematic to use the sum-of-squares type tests, since summarizing information through averaging will weaken the signals carried in the few securities with strong signals. Therefore, it is desirable to develop a test that has good power against sparse alternatives.

    Recently, the maximum-type tests are wildly used in high dimensional testing problems. \citet{TonyCai2014TwosampleTO} proposed three-types maximum-type tests for high dimensional two-sample location problems. \citet{chang2017} consider testing for high-dimensional white noise using maximum cross-correlations.
    In the traditional linear factor pricing models, many studies also show that the maximum-type test statistic performs very well under the sparse alternatives, such as \citet{Gungor2013TestingLF,Feng2022AsymptoticIO,Yu2023PE}. So, we construct an maximum-type test statistic for the time-varying factor pricing models and establish its theoretical properties. Theoretical results and simulation studies show the proposed maximum-type test statistic also has good performance under sparse alternatives.

    It's worth notable that the underlying truth is usually unknown in real applications, as whether it is dense or sparse depends on the properties of the involved securities. Recently, many literatures showed that the maximum-type test statistic is asymptotically independent with the sum-type test statistic and proposed a corresponding combination test in many high dimensional problems, such as \citet{he2021}, \citet{Feng2022AsymptoticIO} for high dimensional mean testing problems, \citet{feng2022a} for cross-sectional independence test in high dimensinal panel data models, \citet{yu2022jasa} for testing of high dimensional covariance matrix, \citet{wang2023} for high dimensional change point inference. In the time-varying factor pricing model, we also demonstrate the asymptotic independence between the maximum-type statistic and the sum-type statistic proposed in \citet{Ma2018TestingAI}. Then, we construct an adaptive testing procedure by combining the information from these two statistics, which would yield high power against various alternatives. We demonstrate the advantages of the proposed tests over existing methods through extensive Monte Carlo experiments and two empirical applications.
	
    The rest of the paper is organized as follows. In Section \ref{sec:max}, we introduce the maximum-type statistic and establish its theoretical properties. In Section \ref{sec:adp}, we integrate the proposed maximum-type test with an existing sum-type test to obtain an adaptive test. Monte Carlo experiment results are presented in Section \ref{sec:sim} to evaluate the finite sample performance of the proposed tests in comparison with the main competitors. Two empirical applications to the security return data from the Chinese and the U.S. financial markets are presented in Section \ref{sec:app}. Finally, we conclude the paper with some discussions in Section \ref{sec:conc} and relegate the technical proofs to the Appendix.
	
    Finally we introduce some notation. For any vector $\bm{z}=(z_1,...,z_m)'\in\mathbb{R}^m$, let $||\bm{z}||=(\sum_{i=1}^m z_i^2)^{1/2}$ and $||\bm{z}||_\infty=\max_{1\le i\le m}|z_i|$. Let $\bm{1}_m$ be the $m\times1$ vector of ones. For any positive numbers $a_n$ and $b_n$, let $a_n\ll b_n$ denote $a_n b_n^{-1}=o(1)$, $a_n\sim b_n$ denote $\lim_{n\to\infty}a_n b_n^{-1}=1$, and $a_n\asymp b_n$ denote $\lim_{n\to\infty}a_n b_n^{-1}=c$ for some finite positive constant $c$. For an $m\times n$ matrix $\bm{A}=(a_{i j})$, let $\tr(\bm{A})$ denote the trace of $\bm{A}$, $P_{\bm{A}}=\bm{A}(\bm{A}'\bm{A})^{-1}\bm{A}'$ and $M_{\bm{A}}=\bm{I}_m-P_{\bm{A}}$, where $\bm{I}_m$ is the $m\times m$ identity matrix. Moreover, denote $||\bm{A}||=\max_{\bm{a}\in\mathbb{R}^n,||\bm{a}||=1}||\bm{A a}||$ and $||\bm{A}||_\infty=\max_{1\le i\le m}\sum_{j=1}^n|a_{i j}|$. For any symmetric matrix $\bm{A}\in\mathbb{R}^{n\times n}$, let $\lambda_{\min}(\bm{A})$ and $\lambda_{\max}(\bm{A})$ denote the smallest and largest eigenvalues of $\bm{A}$, respectively. $(N,T)\to \infty$ denotes that $N$ and $T$ go to infinity jointly. The operators $\stackrel{d}{\to}$ and $\stackrel{p}{\to}$ denote convergence in distribution and in probability, respectively.

\section{The Maximum-type Test}
\label{sec:max}
\subsection{Econometric Framework and B-spline Approximation}
    We consider the following conditional time-varying factor model under high-dimensional case,
    \begin{equation}\label{equ.2}
    R_{it}=\alpha_{it}+\bm{\beta}_{it}'\bm{f}_t+e_{it}=\alpha_{it}+\sum_{j=1}^{d}\beta_{ijt}f_{jt}+e_{it}, \quad i=1,...,N; t=1,...,T,
    \end{equation}
    where $R_{it}$ denotes the excess returns of the $i$-th asset at time $t$, $\alpha_{it}$ is the conditional alpha of asset $i$ at time $t$, $\bm{\beta}_{it}=(\beta_{i1t},...,\beta_{idt})'$ is a $d\times1$ vector of time-variant factor loadings, $\bm{f}_t=(f_{it},...,f_{dt})'$ stands for the risk premium on $d$-dimensional tradable systematic risks at time $t$, and $e_{it}$ is the idiosyncratic error term with $\cov(\bm{e}_t)=\bm{\Sigma}=(\sigma_{ij})_{N\times N}$, where $\bm{e}_t=(e_{1t},...,e_{Nt})'\in \mathbb{R}^{N\times1}$.

    To identify the parameters, we follow \citet{Li2011TestingCF} and others in imposing a smoothness condition. That is we assume that the sequence of alphas and betas are generated from two smoothing functions of time such that $\alpha_{it}=\alpha_i(t/T)$ and $\beta_{ijt}=\beta_{ij}(t/T)$. Accordingly, we can rewrite equation (\ref{equ.2}) as
    \begin{equation}\label{equ.3}
    R_{it}=\delta_i^0+\delta_i(t/T)+\sum_{j=1}^{d}\beta_{ij}(t/T)f_{jt}+e_{it},
    \end{equation}
    where $\delta_i^0=T^{-1}\sum_{t=1}^{T}\alpha_{it}$ and $\delta_i(t/T)=\alpha_i(t/T)-T^{-1}\sum_{i=1}^{T}\alpha_i(t/T)$. To test whether the average pricing error for $N$ assets across all time periods is equal to zero, we then focus on testing
    \begin{equation}\label{equ.4}
    H_0:\delta_i^0=0 \text{ for all }i=1,...,N \quad \text{v.s.}\quad  H_1:\delta_i^0\ne 0 \text{ for some }i=1,...,N.
    \end{equation}

    For testing hypothesis in equation (\ref{equ.4}), we employ the polynomial spline approach to estimating the unknown parameters $\delta_i(t/T)$ and $\beta_{ij}(t/T)$ under $H_0$. Consider $p$ interior knots on $[0,1]$, $0=\ell_0<\ell_1<...<\ell_p<\ell_{p+1}=1$ that satisfy
    \begin{equation*}
    \max_{0\le i\le p}|\ell_{i+1}-\ell_i|/\min_{0\le l\le p}|\ell_{i+1}-\ell_i|\le c
    \end{equation*}
    for some finite positive constant $c$, where $p=p(N,T)\to\infty$ as $(N,T)\to\infty$. For any $t$, define its location as $l(t)$ satisfying $\ell_{l(t)}\le t/T<\ell_{l(t)+1}$. Consider the space of polynomial splines of order $q$ and denote the normalized B spline basis of this space as $\bm{B}(t/T) = \{B_1(t/T),...,B_L(t/T)\}'$,
    where $L = p + q$ \citep{Boor1978APG}. To estimate $\delta_i(\cdot)$, we consider the centered spline basis functions, $\widetilde{B}_l(t/T)=B_l(t/T)-T^{-1}\sum_{t=1}^{T}B_l(t/T)$, and denote $\widetilde{\bm{B}}(t/T)=\{\widetilde{B}_1(t/T),...,\widetilde{B}_L(t/T)\}'$. Then, the unknown functions $\delta_i(\cdot)$ and $\beta_{ij}(\cdot)$ can be well approximated by the B-spline functions \citep{Schumaker1981SplineFB} such that
    \begin{equation*}
    \delta_i(t/T)\approx \bm{\lambda}_{i0}'\widetilde{\bm{B}}(t/T) \quad \text{and}\quad \beta_{ij}(t/T)\approx \bm{\lambda}_{ij}'\bm{B}(t/T),
    \end{equation*}
    where $\bm{\lambda}_{ij}\in \mathbb{R}^{L\times1}$, $j=0,1,...,d$ are the coefficients of the B-spline functions.

    Denote $\bm{R}_{i.}=(R_{i1},...,R_{iT})'\in \mathbb{R}^{T\times1}$, $\bm{e}_{i.}=(e_{i1},...,e_{iT})'\in \mathbb{R}^{T\times1}$, $\bm{\lambda}_i=(\bm{\lambda}_{ij}',0\leq j\leq d)'\in \mathbb{R}^{(1+d)L\times1}$, $\bm{Z}_t=\{\widetilde{\bm{B}}(t/T)',f_{jt}\bm{B}(t/T)',1\leq j\leq d\}'\in \mathbb{R}^{(1+d)L\times1}$ and $\bm{Z}=(\bm{Z}_1,...,\bm{Z}_T)'\in \mathbb{R}^{T\times (1+d)L}$. We then choose estimators $\widehat{\bm{\lambda}}_i=(\widehat{\bm{\lambda}}_{ij}',0\leq j\leq d)'$ to minimize the following sum of squared residuals:
    \begin{equation*}
    \widehat{\bm{\lambda}}_i=\arg\min_{\bm{\lambda}_i}||\bm{R}_{i.}-\bm{Z}\bm{\lambda}_i||^2=(\bm{Z}'\bm{Z})^{-1}\bm{Z}'\bm{R}_{i.}.
    \end{equation*}
    Consequently, the resulting residuals are
    \begin{equation}\label{equ.5}
    \widehat{\bm{e}}_{i.}=\bm{R}_{i.}-\bm{Z}\widehat{\bm{\lambda}}_i=\bm{M_ZR}_{i.}.
    \end{equation}

\subsection{The Proposed Test}
    Under $H_0$, it can be shown that $\delta_i(t/T)-\widehat{\delta}_i(t/T) \stackrel{p}{\to}0$ and $\beta_{ij}(t/T)-\widehat{\beta}_{ij}(t/T) \stackrel{p}{\to}0$. This motivates us to consider the max-type test, with the test statistic constructed as,
    \begin{equation}\label{equ.6}
    M_{NT}=\max_{1\le i\le N}t_i^2,
    \end{equation}
    where
    \begin{equation}\label{equ.7}
    t_i^2=T^{-1}\widehat{\sigma}_{ii}^{-1}(\widehat{\bm{e}}_{i.}'\bm{1}_T)^2,
    \end{equation}
    with $\widehat{\sigma}_{ij}=\widehat{\bm{e}}_{i.}'\widehat{\bm{e}}_{j.}/(T-d-1)$.

    We will establish that when $(N,T)\to\infty$, $M_{NT}-2\log(N) +\log\{\log(N)\}$ has a type \uppercase\expandafter{\romannumeral1} extreme value distribution. To proceed, we first introduce some assumptions.

    Let $\mathcal{H}_r$ denote the collection of all functions on $[0,1]$ such that the $m$-th order derivative satisfies the H\"older condition of order $n$ with $r=m+n$, i.e. there exists a constant $ C\in(0,\infty)$ such that for each $f\in\mathcal{H}_r$,
    \begin{equation*}
    |f^{(m)}(x_1)-f^{(m)}(x_2)|\le C|x_1-x_2|^n
    \end{equation*}
    for any $0\le x_1,x_2\le 1$. Let $\mathcal{F}_{NT,t}$ denote the $\sigma$-algebra generated from $\{\bm{f},\{e_{it},e_{i,t-1},...\}_{i=1}^N\}$, where $\bm{f}=(\bm{f}_1',...,\bm{f}_T')'$. Let $\bm{e}_{-t}=\{e_{i1},...,e_{i,t-1},e_{i,t+1},...,e_{iT}\}_{i=1}^N$. Let $\bm{D}$ denote the diagonal matrix of $\bm{\Sigma}$ and $\bm{R}=(r_{ij})=\bm{D}^{-1/2}\bm{\Sigma}\bm{D}^{-1/2}$ denote the correlation matrix.
    \begin{assumption}
    	\label{A2.1}
    	$\delta_i(\cdot)\in \mathcal{H}_r$ and $\beta_{ij}(\cdot)\in \mathcal{H}_r$ for some $r>3/2$.
    \end{assumption}
    \begin{assumption}
    	\label{A2.2}
    	(\romannumeral1) There exist constants $0<c_0<C_0<\infty$ such that
    	\begin{equation*}
    	c_0\leq \lambda_{\min}(\mathbb{E}\{(1,\bm{f}_t')'(1,\bm{f}_t')\})
    	\leq \lambda_{\max}(\mathbb{E}\{(1,\bm{f}_t')'(1,\bm{f}_t')\})\leq C_0
    	\end{equation*}
    	holds uniformly for $t\in [1,T]$;
    	(\romannumeral2) There exist a constant $0<K<\infty$ such that $\mathbb{E}||\bm{f}_t||^{4(2+\kappa)}\leq K$ for some $\kappa>0$;
    	(\romannumeral3) The process $\{\bm{f}_t\}_{t\ge 1}$ is strong mixing with mixing coefficient $\alpha(\cdot)$ satisfying $\sum_{k=0}^\infty\alpha(k)^{\kappa/(2+\kappa)}<\infty$;
    	(\romannumeral4) $\{\bm{f}_t\}_{t=1}^T$ and $\{\bm{e}_t\}_{t=1}^T$ are independent.
    \end{assumption}
    \begin{assumption}
    	\label{A2.3}
    	(\romannumeral1) $\bm{e}_1,...,\bm{e}_T$ are independently and identically distributed with for each $i$, $\mathbb{E}(e_{it}|\mathcal{F}_{NT,t-1})=0$ and $\mathbb{E}(\bm{e}_t\bm{e}_t'|\bm{e}_{-t})=\bm{\Sigma}$, where $\bm{\Sigma}$ is positive definite and $\sigma_{ii}\in (0,\infty)$ for every $1\leq i\leq N$;
    	(\romannumeral2)
    	$e_{it}$'s have sub-Gaussian-type tails, i.e. there exist $\eta>0$ and $K>0$ such that $\mathbb{E}\{\exp(\eta e_{it}^2/\sigma_{ii})\}\leq K$ for $1\leq i\leq N$.
    \end{assumption}
    \begin{assumption}
    	\label{A2.4}
    	(\romannumeral1) $TL^{-2r}=o(1)$;
    	(\romannumeral2) There exists $c_1>0$ such that $c_1^{-1}\leq \lambda_{\min}(\bm{R})\leq \lambda_{\max}(\bm{R})\leq c_1$;
    	(\romannumeral3) There exists $k>0$ such that $\max_{1\leq i,j\leq N}|r_{ij}|\leq k<1$.
    \end{assumption}
    Assumption \ref{A2.1} is the common smoothness assumption on the unknown functions \citep{He1996BivariateTB}. Assumption \ref{A2.2} is the typical condition for the regression design matrix. Specially, Assumption \ref{A2.2}(\romannumeral1) follows Condition (C2) in \citet{Wang2008VariableSI}, and Assumption \ref{A2.2}(\romannumeral4) follows from Assumption 3.1(\romannumeral2) in \citet{Fan2011HighDC}. Moreover, Assumptions \ref{A2.2}(\romannumeral2)-(\romannumeral3) weaken Assumption 3.2 and 3.3(\romannumeral2) in \citet{Fan2011HighDC}. Assumption \ref{A2.3} is the moment conditions on the distribution of the error terms. Following \citet{Pesaran2012TestingCW} and others, we assume that $\bm{e}_t$'s are independently and identically distributed. Moreover, Assumption \ref{A2.3}(\romannumeral1) contains the assumption for a martingale difference sequence, and the homogeneity assumption on the covariance of the error terms. Assumption \ref{A2.3}(\romannumeral2) is equivalent to Condition (C6) in \citet{TonyCai2014TwosampleTO}, which allows the theoretical results to hold for error distributions more general than the Gaussian type. Assumption \ref{A2.4}(\romannumeral1) is the condition on the number of spline basis functions $L$. Assumption \ref{A2.4}(\romannumeral2) is a common assumption on the eigenvalues in the high dimensional setting, which is equivalent to Assumption 4.1(\romannumeral1) in \citet{Fan2013PowerEI} when $\sigma_{ii}$'s are bounded. Assumption \ref{A2.4}(\romannumeral3) is also mild. For example, if $\max_{1\le i<j\le N}|r_{ij}|=1$, then $\bm{\Sigma}$ is singular.

    We now state our first main result, which is about the asymptotic property of $M_{NT}$.
    \begin{theorem}\label{T2.1}
    	Suppose that Assumptions \ref{A2.1}-\ref{A2.4} hold. Assuming $L^3T^{-1}=o(1)$ and $\log(N)=o(T^{1/4})$, we have as $(N,T)\to \infty$, for any $x\in\mathbb{R}$,
    	\begin{equation*}
    	\mathbb{P}_{H_0}\left(M_{NT}-2\log(N)+\log\{\log (N)\}\leq x \right) \to F(x)\equiv\exp\left\lbrace-\frac{1}{\sqrt{\pi}}\exp\left(-\frac{x}{2} \right)  \right\rbrace.
    	\end{equation*}	
    \end{theorem}
    Here, $\mathbb{P}_{H_0}$ denotes the probability measure under the null hypothesis $H_0$. According to the limiting null distribution derived in, we can easily obtain the $p$-value associated with $M_{NT}$, namely,
    \begin{equation*}
    p_M=1-F(M_{NT}-2\log(N) +\log\{\log(N)\}).
    \end{equation*}
    If the $p$-value is below
    some pre-specified significant level, say $\gamma\in(0,1)$, then we rejected the null hypothesis that the traded factors are sufficient to price all assets. Next, we turn to analyze the  power of the maximum-type testing procedure.

    \begin{proposition}\label{P2.1}
    	Suppose Assumptions \ref{A2.1}-\ref{A2.4} hold. Assuming  $L^3T^{-1}=o(1)$ and $\log(N)=o(T^{1/4})$, we have as $(N,T)\to\infty$,
    	\begin{equation*}
    	\inf_{\bm{\alpha}\in\mathcal{A}(C)}\mathbb{P}(p_M<\gamma)\to 1,
    	\end{equation*}
    	for some large enough constant $C>0$, where
    	\begin{equation*}
    	\mathcal{A}(C)=\left\lbrace \bm{\alpha}=(\alpha_{it})_{N\times T}:\max_{1\le i\le N}\left|\frac{1}{T}\sum_{t=1}^T\alpha_{it} \right|\ge C\sqrt{\frac{\log(N)}{T}}  \right\rbrace.
    	\end{equation*}
    	
    \end{proposition}
    Proposition \ref{P2.1} shows that the proposed max-type test is effective in detecting sparse alternatives.

\section{The Adaptive Test}
\label{sec:adp}
\subsection{The Existing sum-type Test}
    To detect dense alternatives, we consider the sum-type test statistic proposed by \citet{Ma2018TestingAI}. To wit,
    \begin{equation}\label{equ.8}
    S_{NT}=N^{-1}T^{-1}\sum_{i=1}^N(\widehat{\bm{e}}_{i.}'\bm{1}_T)^2,
    \end{equation}
    whose mean and variance under $H_0$ are given below,
    \begin{equation*}
    \mu_{NT}^0=N^{-1}T^{-1}\sum_{i=1}^N\sum_{t=1}^T\mathbb{E}(e_{it}^2)\mathbb{E}(h_t^2) \text{ and } \sigma_{NT}^2=2N^{-2}T^{-2}\tr(\bm{\Sigma}^2)\sum_{1\le t\ne s\le T}\mathbb{E}(h_t^2h_s^2),
    \end{equation*}
    where $\bm{h}=(h_1,...,h_T)'=M_{\bm{Z}}\bm{1}_T$. Let $\mu_{NT}=N^{-1}T^{-1}\sum_{i=1}^N\sum_{t=1}^Te_{it}^2h_t^2$ be an empirical approximation of the mean. For standardization, they considered the consistent estimator of $\mu_{NT}^0$ as
    \begin{equation}\label{equ.9}
    \widehat{\mu}_{NT}=N^{-1}T^{-1}\sum_{i=1}^N\sum_{t=1}^T\widehat{e}_{it}^2h_t^2.
    \end{equation}
    As for $\sigma_{NT}^2$, following \citet{Lan2013TestingCI}, they proposed the following estimator
    \begin{equation}\label{equ.10}
    \widehat{\sigma}_{NT}^2=2N^{-2}T^{-2}\widehat{\tr(\bm{\Sigma}^2)}\sum_{1\le t\ne s\le T}h_t^2h_s^2,
    \end{equation}
    where  $\widehat{\tr(\bm{\Sigma}^2)}=T^2\{T+(1+d)L-1\}^{-1}\{T-(1+d)L\}^{-1}[\tr(\widehat{\bm{\Sigma}}^2)-\tr^2(\widehat{\bm{\Sigma}})/\{T-(1+d)L\}]$, with $\widehat{\bm{\Sigma}}=T^{-1}\sum_{t=1}^T(\widehat{\bm{e}}_t-\overline{\bm{e}})(\widehat{\bm{e}}_t-\overline{\bm{e}})^T$ and $\overline{\bm{e}}=T^{-1}\sum_{t=1}^T\widehat{\bm{e}}_t$.

    The following Lemma restates the asymptotic null distribution of $S_{NT}$ derived in \citet{Ma2018TestingAI}.
    \begin{assumption}
    	\label{A3.1}
    	(\romannumeral1) $TL^{-2r}N\tr^{-1/2}(\bm{\Sigma}^2)=o(1)$;
    	(\romannumeral2) $\tr^{-1/2}(\bm{\Sigma}^2)\max_i\sum_{j=1}^N|\sigma_{ij}|=o(1)$.	
    \end{assumption}
    \begin{assumption}
    	\label{A3.2}
    	(\romannumeral1) $tr(\bm{\Sigma}^4)=o\{\tr^2(\bm{\Sigma}^2)\}$ as $N\to \infty$;
    	(\romannumeral2) $T^{-2}\sum_{t=1}^T\mathbb(\bm{e}_t'\bm{\Sigma} \bm{e}_t)^2=o\{\tr^2(\bm{\Sigma}^2)\}$;
    	(\romannumeral3) $T^{-1+\varrho}N^{1+\varrho}L\tr^{-1/2}(\bm{\Sigma}^2)=O(1)$ for an arbitratily small $\varrho>0$.
    \end{assumption}
    \begin{lemma}\label{T3.1}
    	Suppose Assumptions \ref{A2.1}-\ref{A2.3} and \ref{A3.1}-\ref{A3.2} hold. Assume that  $L^3T^{-1}=o(1)$ and $L^rT^{3/2}=o(1)$. Then, under $H_0$, as $(N,T)\to \infty$,
    	\begin{equation*}
    	(S_{NT}-\widehat{\mu}_{NT} )/\widehat{\sigma}_{NT}\stackrel{d}{\to}\mathcal{N}(0,1).
    	\end{equation*}	
    \end{lemma}
    By Lemma \ref{T3.1}, the $p$-value associated with $S_{NT}$ is
    \begin{equation*}
    p_S=1-\Phi\left( (S_{NT}-\widehat{\mu}_{NT} )/\widehat{\sigma}_{NT}\right),
    \end{equation*}
    where $\Phi(\cdot)$ is the cumulative distribution function (CDF) of $\mathcal{N}(0,1)$. Again, small values of $p_S$ indicate rejecting the null hypothesis.

\subsection{Adaptive Strategy}
    In practice, we seldom know whether the vector of intercepts is sparse or dense. In order to adapt to various alternative behaviors, we combine the maximum- and sum-type testing procedures. The key message is that both test statistics are asymptotically independent under some mild conditions if $H_0$ holds.
    \begin{theorem}\label{T3.2}
    	Suppose Assumptions \ref{A2.1}-\ref{A2.4} and \ref{A3.1}-\ref{A3.2} hold. Assuming $L^3T^{-1}=o(1)$, $\log(N)=o(T^{1/4})$ and $L^rT^{3/2}=o(1)$, we have as $(N,T)\to\infty$, under $H_0$,
    	\begin{equation*}
    	\mathbb{P}_{H_0}\left( \frac{S_{NT}-\widehat{\mu}_{NT}}{\widehat{\sigma}_{NT}} \leq x, M_{NT}-2\log(N)+\log\{\log(N)\}\leq y\right) \to \Phi(x)F(y).
    	\end{equation*}
    \end{theorem}
    According to Theorem \ref{T3.2}, we suggest combining the corresponding $p$-values by using Fisher’s method \citep{Littell1971AsymptoticOO}, to wit,
    \begin{equation*}
    p_{adp}=1-G(-2\{\log(p_M)+\log(p_S)\}),
    \end{equation*}
    where $G(\cdot)$ is the CDF of the chi-squared distribution with 4 degrees of freedom. The rationality is that $-2\{\log(p_M)+\log(p_S)\}$ converges in distribution to $G$ under $H_0$ due to Theorem \ref{T3.2}. If the final $p$-value is less than some pre-specified significant level $\gamma\in(0,1)$, then we reject $H_0$.

    Next, we analyze the power of the adaptive testing procedure. We consider the following sequence of alternative hypotheses, to wit,
    \begin{equation}\label{equ.11}
    H_{1,NT}: ||\bm{\delta}^0||_0=o[N/\log^2\{\log(N)\}] \text{ and }
    ||\bm{\delta}^0||=O\{T^{-1/2}\tr^{1/4}(\bm{\Sigma}^2)\},
    \end{equation}
    where $\bm{\delta}^0=(\delta_1^0,...,\delta_N^0)^T$ with $\delta_i^0=T^{-1}\sum_{t=1}^T\alpha_{it}$. In fact, the asymptotic independence between the maximum-type and sum-type statistics still hold under the hypotheses given in equation (\ref{equ.11}).
    \begin{theorem}\label{T3.3}
    	Under the same condition as Theorem \ref{T3.2}, we have as $(N,T)\to\infty$, under $H_{1,NT}$,
    	\begin{align*}
    	&\mathbb{P}\left( \frac{S_{NT}-\widehat{\mu}_{NT}}{\widehat{\sigma}_{NT}} \leq x, M_{NT}-2\log(N)+\log\{\log(N)\}\leq y\right) \\
    	&\qquad \to\mathbb{P}\left( \frac{S_{NT}-\widehat{\mu}_{NT}}{\widehat{\sigma}_{NT}} \leq x\right)\mathbb{P}\left( M_{NT}-2\log(N)+\log\{\log(N)\}\leq y\right).
    	\end{align*}
    \end{theorem}

    Simulation studies show that the power of Fisher's p-value combination-based test would be comparable to that of the test based on $\min\{p_{M},p_{S}\}$ (referred to as the minimal p-value combination), say $\beta_{M\wedge S, \alpha}=P(\min\{{\rm p}_M,{\rm p}_S\}\leq 1-\sqrt{1-\alpha})$. Obviously,
    \begin{align}\label{power_H1}
    \beta_{M\wedge S, \alpha} &\ge P(\min\{{\rm p}_M,{\rm p}_S\}\leq \alpha/2)\nonumber\\
    &= \beta_{M,\alpha/2}+\beta_{S,\alpha/2}-P({\rm p}_M\leq \alpha/2, {\rm p}_S\leq \alpha/2)\nonumber\\
    &\ge \max\{\beta_{M,\alpha/2},\beta_{S,\alpha/2}\}.
    \end{align}
    On the other hand, under $H_{1}$ in \eqref{equ.11}, we have
    \begin{align}\label{power_H1np}
    \beta_{M\wedge S, \alpha} \ge \beta_{M,\alpha/2}+\beta_{S,\alpha/2}-\beta_{M,\alpha/2}\beta_{S,\alpha/2}+o(1),
    \end{align}
    due to the asymptotic independence entailed by Theorem \ref{T3.3}. For a small $\alpha$, the difference between $\beta_{M,\alpha}$ and $\beta_{M,\alpha/2}$ should be small, and the same fact applies to $\beta_{S,\alpha}$. Consequently, by \eqref{power_H1}--\eqref{power_H1np}, the power of the adaptive test would be no smaller than or even significantly larger than that of either max-type or sum-type test.

\section{Monte Carlo Experiments}
\label{sec:sim}
\subsection{Experiment Settings}
    We conduct Monte Carlo experiments to evaluate the performance of the proposed tests, and investigate the relationship between the power and sparsity levels or signal sizes. For the conditional model discussed in Section \ref{sec:max}, we consider the following two examples.
    \\~\\
    \textbf{Example 1.} Following \citet{Ma2018TestingAI}, we generate data from the well known conditional CAPM:

    \begin{equation}\label{equ.12}
    R_{it}=\alpha_{it}+\beta_{it}f_t+e_{it},\quad i=1,...N; t=1,...,T,
    \end{equation}
    where $f_t$ is the Market factor. Specifically, we generate the factor from the following AR(1)-GARCH(1,1) processes:
    \begin{equation*}
    f_t=0.34+0.05(f_{t-1}-0.34)+h_t^{1/2}\zeta_t,
    \end{equation*}
    where $\zeta_t$ is simulated from a standard normal distribution, the variance terms $h_t$ follows from the process
    \begin{equation*}
    h_t=0.23+0.67h_{t-1}+0.13h_{t-1}\zeta_{t-1}^2,
    \end{equation*}
    and the above coefficients are obtained by fitting the model to the U.S. stock market data.

    As for the error terms, $\bm{e}_t=(e_{1t},...,e_{Nt})^T\in\mathbb{R}^N$ is generated from $\bm{e}_t\sim\bm{\Sigma}^{1/2}\bm{z}_t$, where $\bm{z}_t$ has $N$ i.i.d. entries of $N(0,1)$ and $\exp(1)$, respectively. Following \citet{Fan2011HighDC}, $\bm{\Sigma}=(\sigma_{i_1i_2})\in\mathbb{R}^{N\times N}$ with $\sigma_{i_1i_2}=0.5^{|i_1-i_2|}$, which implies that $e_{i_1t}$ and $e_{i_2t}$ are approximately uncorrelated when the difference $|i_1-i_2|$ is sufficiently large.

    To assess the robustness of the proposed test for the random factor loadings, we set the conditional factor loadings to be $\beta_{it}=1+0.5\xi_t$ for $i=1,...,N$ and $t=1,...,T$, where the unobservable state variable $\xi_t$ follows an AR(1)-ARCH(1) process, $\xi_t=0.8\xi_{t-1}+u_t$ with $u_t=v_t\varepsilon_t$, $\varepsilon_t\sim N(0,1)$ and $v_t^2=0.1+0.6v_{t-1}^2$.

    Finally, we consider the conditional alphas. We set $\bm{\alpha}=\bm{0}$ under the null hypothesis. For the alternative hypothesis, we set $\alpha_{it}=\alpha_it/T$ for $i\in S\subset\{1,...,N\} $ and $t=1,...,T$, where each element in $S$ is uniformly and randomly drawn from $\{1,...,N\}$ with $|S|=s$, and $\alpha_i$'s are independently generated from $U(0,c\sqrt{\log(N)/(Ts)})$. And we keep the remaining $\alpha_i$ with $i\notin S$ zero. To examine how the power changes accordingly, we let the signal strength $c$ and sparsity level $s$ vary.

    The above processes are simulated over the period $t=-24,...,0,1,...,T$ with the initial values $R_{i,-25}=0$, $f_{-25}=0$, $h_{-25}=1$, $z_{-25}=0$ and $v_{-25}^2=1$. To offset the start-up effects, we drop the first $25$ simulated observations and use $t=1,...,T$ for our final experiments.
    \\~\\
    \textbf{Example 2.}
    To mimic the commonly used conditional FF model, where the factors $\bm{f}_t$ have strong serial correlation and heterogeneous variance, we generate $R_{it}$ according to the following model with $d=3$:
    \begin{equation}\label{equ.13}
    R_{it}=\alpha_{it}+\sum_{j=1}^d\beta_{ijt}f_{jt}+e_{it},\quad i=1,...N; t=1,...,T,
    \end{equation}
    where $f_{1t}$, $f_{2t}$ and $f_{3t}$ are the market factor, SMB and HML, respectively. These factors are correspondingly simulated from the following AR(1)-GARCH(1,1) processes,
    \begin{align*}
    &\text{Market factor: }f_{1t}=0.34+0.05(f_{1,t-1}-0.34)+h_{1t}^{1/2}\zeta_{1t},\\
    &\text{SMB factor: }f_{2t}=0.04+0.07(f_{2,t-1}-0.04)+h_{2t}^{1/2}\zeta_{2t},\\
    &\text{HML factor: }f_{3t}=0.06+0.04(f_{3,t-1}-0.06)+h_{3t}^{1/2}\zeta_{3t},
    \end{align*}
    where $\zeta_{jt}$'s are simulated from a standard normal distribution, $h_{jt}$'s are generated through the following processes,
    \begin{align*}
    &\text{Market factor: }h_{1t}=0.32+0.67h_{1,t-1}+0.13h_{1,t-1}\zeta_{1,t-1}^2,\\
    &\text{SMB factor: }h_{2t}=0.33+0.51h_{1,t-1}+0.03h_{2,t-1}\zeta_{2,t-1}^2,\\
    &\text{HML factor: }h_{3t}=0.26+0.72h_{1,t-1}+0.05h_{3,t-1}\zeta_{3,t-1}^2,
    \end{align*}
    and all the coefficients are the same as that in \citet{Ma2018TestingAI}.

    The three groups of conditional factor loadings are $\beta_{ijt}=a_j+b_jz_t$ for $i=1,...,N; j=1,2,3$ and $t=1,...,T$ with $(a_1,b_1)=(1,0.5)$, $(a_2,b_2)=(0.1,0.5)$ and $(a_3,b_3)=(0.2,0.4)$, respectively. Additionally, the conditional alphas, the error terms, initial values, and the simulated observations have the same settings as in Example 1.

\subsection{Experiment Results}
    We now present the Monte Carlo experiment results of the proposed maximum-type test and adaptive test, the sum-type test in \citet{Ma2018TestingAI} and the LY test in \citet{Li2011TestingCF}, which are denoted as Max, Adp, Sum and LY, respectively. All the results are based on 1000 replications at the 5\% nominal significance level. In addition, the number of interior knots $n$ is determined by the BIC criterion and the order of B-splines is set as $3$.

        \begin{table}
    	\caption{Empirical size comparison of various tests from Examples 1-2 for testing conditional alphas with a nominal level 5\% and normal or nonnormal errors. \label{tab:size}}
    	\renewcommand\arraystretch{0.6}
    	\resizebox{1.0\linewidth}{!}{
    		\begin{tabular}{lp{1cm} lp{1cm}lp{1cm}
    				lp{1.5cm}lp{1.5cm}lp{1.5cm}lp{1.5cm} lp{1.5cm}lp{1.5cm}lp{1.5cm}lp{1.5cm}}
    			\hline\hline
    			& & & \multicolumn{4}{l}{Normal errors}& \multicolumn{4}{l}{Nonnormal errors}\\
    			\cmidrule(r){4-7} \cmidrule(r){8-11}
    			Example&T& N& Max& Adp& Sum& LY& Max& Adp& Sum& LY\\\hline
    			\multirow{3}{*}{1}&\multirow{3}{*}{500}& 200& 0.025& 0.053& 0.050& 1& 0.040& 0.060& 0.050& 1\\
    			&&500& 0.031& 0.052& 0.049& 1& 0.059& 0.063& 0.047& 1\\
    			&&1000& 0.029& 0.050& 0.048& 1& 0.057& 0.059& 0.045& 1\\\hline\hline
    			& & & \multicolumn{4}{l}{Normal errors}& \multicolumn{4}{l}{Nonnormal errors}\\
    			\cmidrule(r){4-7} \cmidrule(r){8-11}
    			Example&T& N& Max& Adp& Sum& LY& Max& Adp& Sum& LY\\\hline
    			\multirow{3}{*}{2}&\multirow{3}{*}{500}& 200& 0.018& 0.053& 0.053& 1& 0.037& 0.063& 0.053& 1\\
    			&&500& 0.025& 0.057& 0.050& 1& 0.040& 0.054& 0.050& 1\\
    			&&1000& 0.027& 0.055& 0.060& 1& 0.044& 0.064& 0.057& 1\\\hline\hline
    	\end{tabular}}
    \end{table}

    Table \ref{tab:size} summarizes the empirical sizes of four tests under the settings of Example 1-2 over  $N\in \{200, 500, 1000\} $ and $T=500 $. It indicates that Max and Adp perform very well regardless of $N<T$ or $N>T$ and the error distribution being normal or nonnormal, which demonstrates the validity of Theorems \ref{T2.1} and \ref{T3.2}. Sum can roughly maintain the nominal significance level due to Gaussian approximations. In contrast, LY exhibits serious size distortion, since it is not designed for $N > T$. In conclusion, Max, Adp and Sum have a satisfactory performance under the null hypothesis, and LY will be abandoned in the evaluation of power.

    To compare the power performance under different sparsity levels of alphas, we present the empirical power of each test under different $s$'s. For a better visualization, we set the signal strength $c=4$ if the sparsity level $s\in\{4,8,12,16\}$, $c=7$ if $s\in\{18,21,24,27\} $ and $c=10$ if $s\in\{30,60,90,120\}$, which can be roughly regarded as a sparse, a moderately sparse and a dense regime, respectively. To illustrate, Figure \ref{fig:power_sparsity_1} summarizes the results for Example 1 when $(N,T)=(500,500)$. The results of Example 2 are similar. It suggests that under each setting, Adp has the best power performance from the global point of view, since its power performance is always in the first camp in the whole range of $s$. Max outperforms Sum in very sparse case, i.e. $s<12$, while it falls behind Sum in moderately sparse or dense case, i.e. $s>18$.

    Furthermore, we demonstrate the relationship between the signal sizes and the power of the tests. For the alternative hypothesis, we also consider three cases, i.e. the sparsity level $s\in\{2,16,100\} $. In each case, the signal strength $c$ ranges from 0 to 10 with an increment of 0.5. Figure \ref{fig:power_signal_2} presents the results for Example 2 with $(N,T)=(500,500)$. The results of Example 1 are similar. It shows that the type \uppercase\expandafter{\romannumeral1} error rates of Max and Adp are well controlled under $H_0$, and the empirical power of each test increases with the signal strength. Specifically, Max is not powerful when the alternative is dense but becomes more powerful when the alternative gets sparser, while Sum performs much better than Max for dense alternative. Importantly, Adp maintains high power across different signal strengths under all settings of sparsity level.

    In summary, the above experiment results confirm the theoretical conclusions in Section \ref{sec:max} and \ref{sec:adp}. Notably, the proposed adaptive testing procedure is powerful against a wide range of alternatives, and thus advantageous in practice when the true alternative is unknown.

     \begin{figure}
    	\includegraphics[width=1\textwidth]{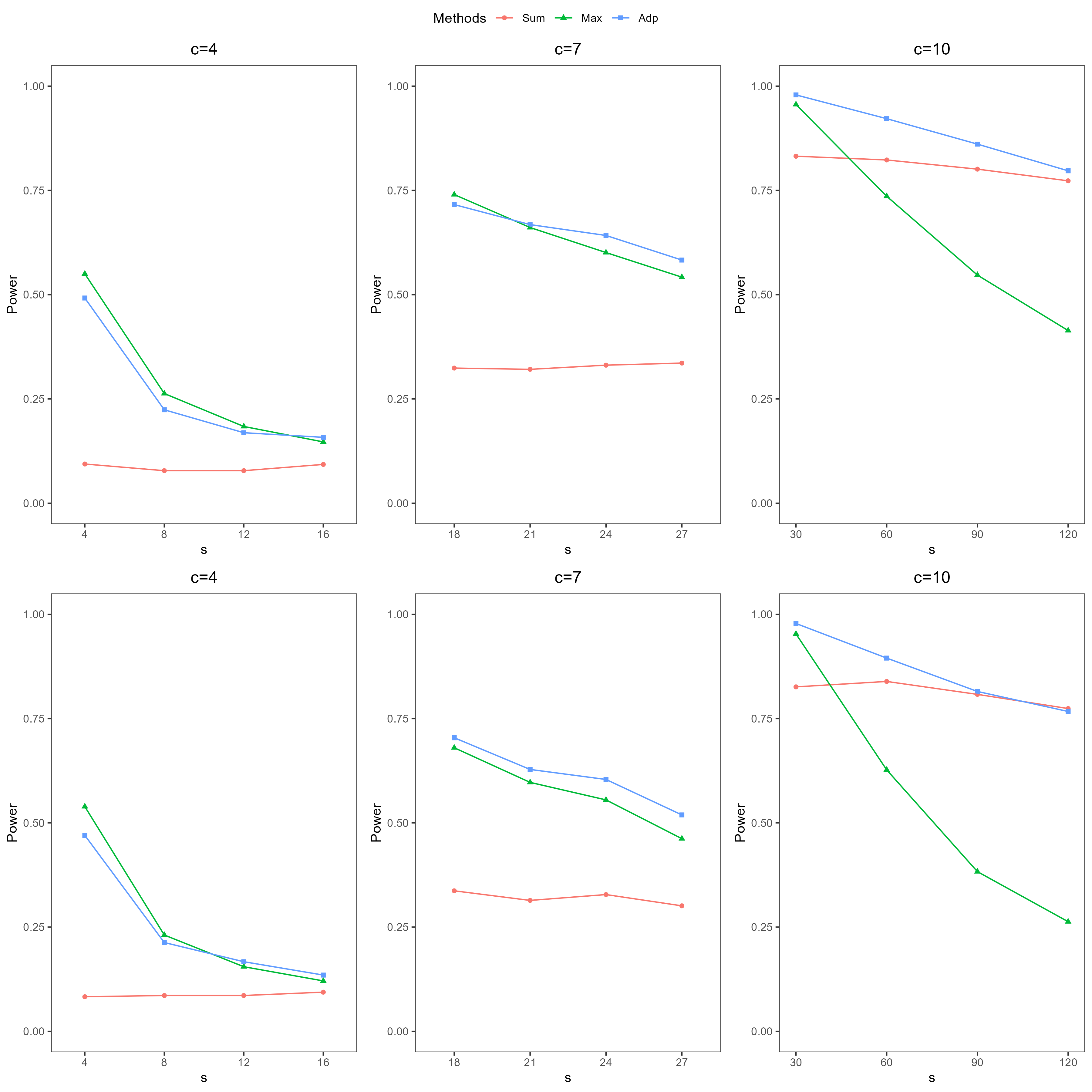}
    	\caption{Power of alpha tests with different sparsity levels for Example 1 over $(N,T)=(500,500)$, where the panels of the first row depict the powers for normal errors, while the panels of the second row depict the powers for nonnormal errors.  \label{fig:power_sparsity_1}}
    \end{figure}

    \begin{figure}
    	\includegraphics[width=1\textwidth]{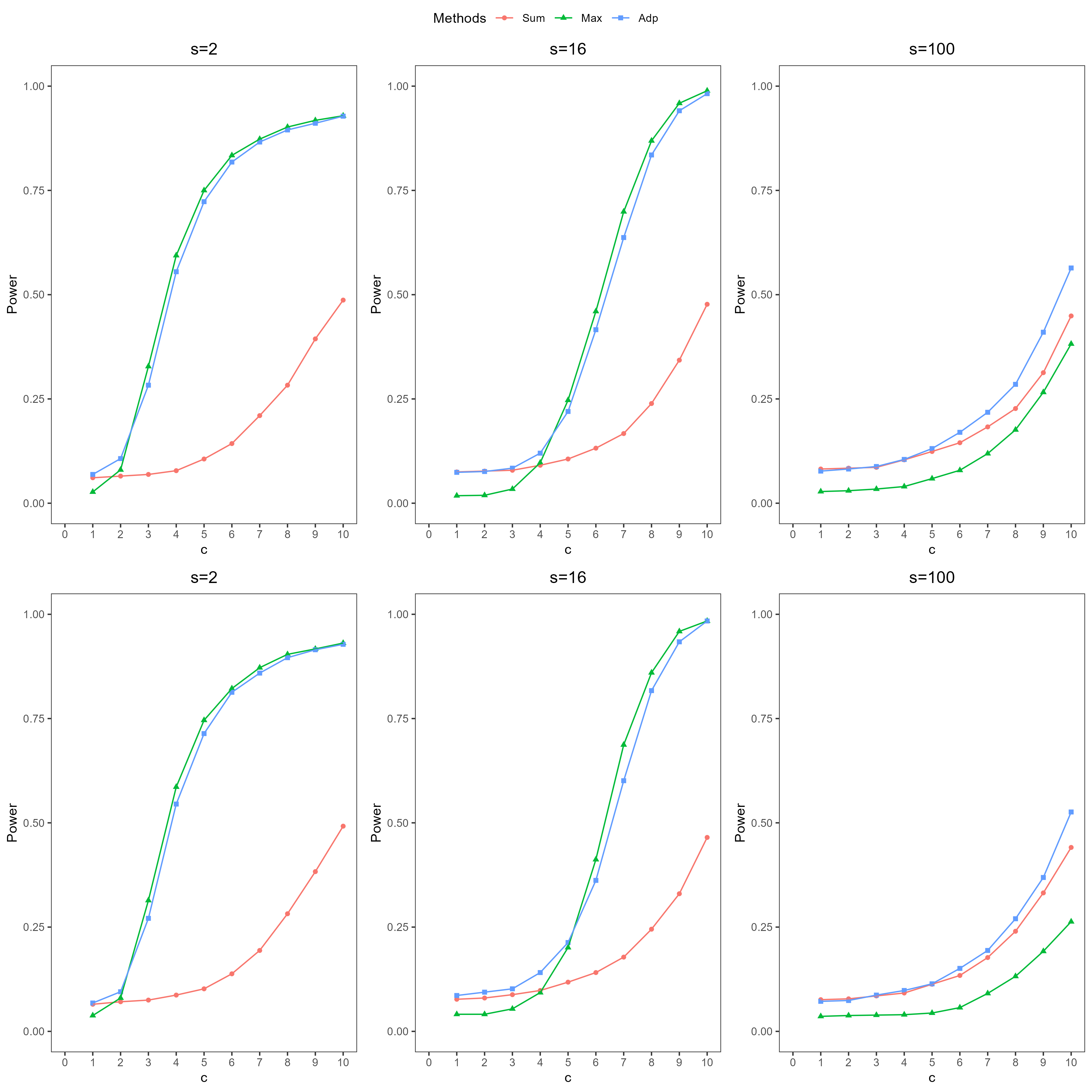}
    	\caption{Power of alpha tests with different signal strengths for Example 2 over $(N,T)=(500,500)$, where the panels of the first row depict the powers for normal errors, while the panels of the second row depict the powers for nonnormal errors. \label{fig:power_signal_2}}
    \end{figure}

\section{Empirical Applications}
\label{sec:app}
    In this section, we employ the proposed tests to analyze the stocks in Chinese and U.S. stock markets. To demonstrate
    the performance of the Max and Adp tests, we compare with the results from their main competitors Sum.

\subsection{Data Description}
    First, we consider the stocks in Chinese stock market. We collected the daily returns of all the stocks in A-shares from 05/12/2021 to 03/09/2022, i.e. $T=200$. After eliminating the stocks with missing observations to avoid analyzing an unbalanced panel, we randomly selected $N=600$ stocks in our final experiment. The time series data on the safe rate of return, and the market factors are calculated according to the corresponding formulas. The risk-free rate $r_{ft}$ is $0.000041$ (from China Asset Management Center). The value-weighted return on all the stocks of Shanghai Stock Exchange and Shenzhen Stock Exchange is used as a proxy for the market return $r_{mt}$. The average return on the three small portfolios minus the average return on the three big portfolios $\text{SMB}_t$, and the average return on two value portfolios minus the average return on two growth portfolios $\text{HML}_t$ are calculated based on the stocks listed on Shanghai Stock Exchange and Shenzhen Stock Exchange. We use $r_{it}$ to denote the return rate of security $i$ on time $t$.

    Next, we consider the stocks in the S\&P 500 index, which is internationally accepted as a leading indicator of the U.S. equities. Similarly, we compiled returns on all the securities that constitute the S\&P 500 index each week over the period from 01/08/2010 to 10/25/2013, i.e. $T=200$. Because the securities that make up the index change over time, we only consider $N=400$ securities that were included in the S\&P 500 index during the entire period. The time series data on the safe rate of return, and the market factors are obtained from Ken French’s data library web page. The one-month U.S. treasury bill rate is chosen as the risk-free rate $r_{ft}$. The value-weighted return on all NYSE, AMEX, and NASDAQ stocks from CRSP is used as a proxy for the market return $r_{mt}$. The average return on the three small portfolios minus the average return on the three big portfolios $\text{SMB}_t$, and the average return on two value portfolios minus the average return on two growth portfolios $\text{HML}_t$ are calculated based on the stocks listed on the NYSE, AMEX and NASDAQ.

\subsection{Conditional Alpha Test}
    We consider the following rolling window procedure with window length $h = 100$. For each $\tau\in\{1,...,T-h\} $, we separately estimate the conditional CAPM and FF model using the data from period $\tau$ to $\tau+h-1$. As a result,
    \begin{align*}
    &\text{CAPM: } R_{it}=r_{it}-r_{ft}=\widehat{\alpha}_{it}+\widehat{\beta}_{it}(r_{mt}-r_{ft})+\widehat{e}_{it},\\
    &\text{FF: } R_{it}=r_{it}-r_{ft}=\widehat{\alpha}_{it}+\widehat{\beta}_{i1t}(r_{mt}-r_{ft})+\widehat{\beta}_{i2t}\text{SMB}_t+\widehat{\beta}_{i3t}\text{HML}_t + \widehat{e}_{it},
    \end{align*}
    for $1\le t\le \tau+h-1$. Based on the estimated residuals $\widehat{e}_{it}$ obtained by separately fitting CAPM and the FF model to the data in each window, we calculate the Max, Sum and Adp test statistics and their corresponding $p$-values. Here, the number of interior knots $n$ is determined via BIC and the order of B-splines is set at 3 for all estimation windows.

    Before applying the involved tests, it is necessary to examine whether alphas and betas are time-varying, given the conclusion of \citet{Li2011TestingCF} that the conditional CAPM and FF are not always superior to their unconditional counterparts. We apply the constant coefficient test (henceforth the CC test) proposed by \citet{Ma2018TestingAI}. The resulting $p$-values are presented in Figure \ref{fig:hist-cc}, which shows that most of the $p$-values of the CC test for 600 (400) stocks in the Chinese (U.S.) dataset is close to $0$, regardless of the model. This provides strong evidence that alphas and betas are indeed time-varying in both the Chinese and U.S. stock markets.

    Now, we apply Max, Adp and Sum tests to the panel data of the securities in the Chinese and U.S. stock markets under the conditional CAPM and FF. Figure \ref{fig:p-values} depicts the $p$-values across the 100 windows. The Box-plot of these $p$-values are presented in Figures \ref{fig:box-p}. In Chinese stock market, we note from Figure \ref{fig:p-values} and Figures \ref{fig:box-p} that the $p$-values of Max and Adp are less than the 5\% significance level for the conditional CAPM, which indicates that the markets are inefficient over these window periods. In contrast, the $p$-values obtained from Sum in the corresponding window periods are greater than 5\%. For the conditional FF, the averaged $p$-values obtained from Max and Adp are smaller than those from Sum. In addition, most of the $p$-values from FF are larger than 5\%, and they are also higher than those from CAPM. Accordingly, FF is better than CAPM in explaining the Chinese stock market. In the U.S. stock market, it is more prominent that the $p$-values are larger than 5\% for both CAPM and FF. This suggests that the U.S. stock market is more efficient than the Chinese stock market.

        \begin{figure*}
    	\centering
    	\begin{tabular}{cc}
    		\includegraphics[width=0.5\textwidth]{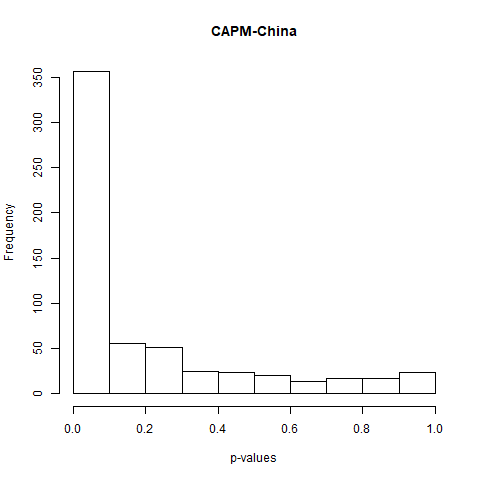}&
    		\includegraphics[width=0.5\textwidth]{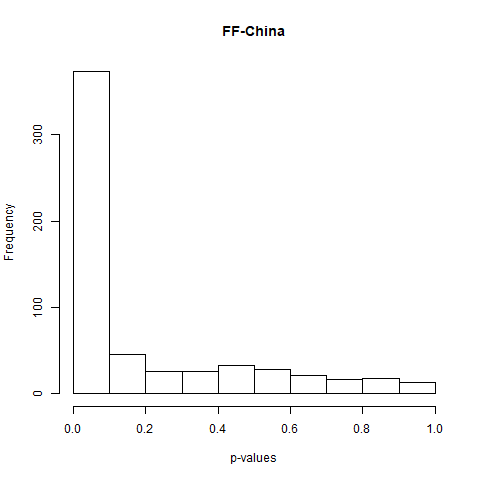}\\
    		\includegraphics[width=0.5\textwidth]{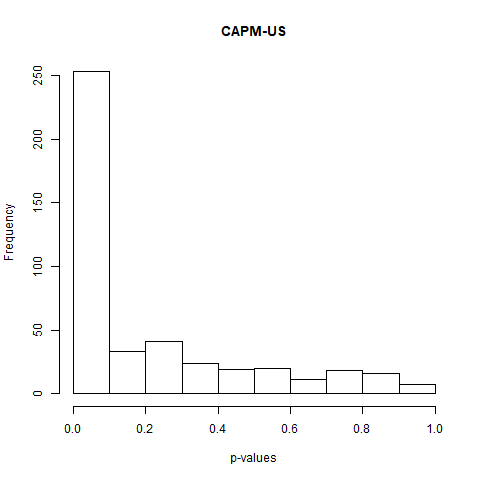}&
    		\includegraphics[width=0.5\textwidth]{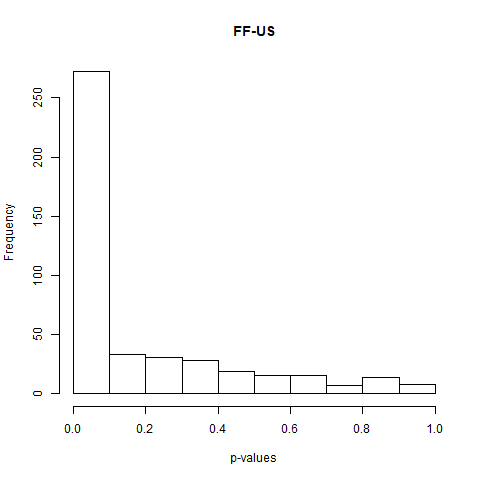} 		
    	\end{tabular}
    	\caption{Histogram of p-values of CC test for Chinese and U.S.'s datasets, respectively.}
    	\label{fig:hist-cc}
    \end{figure*}

    \begin{figure}
    	\includegraphics[width=1\textwidth]{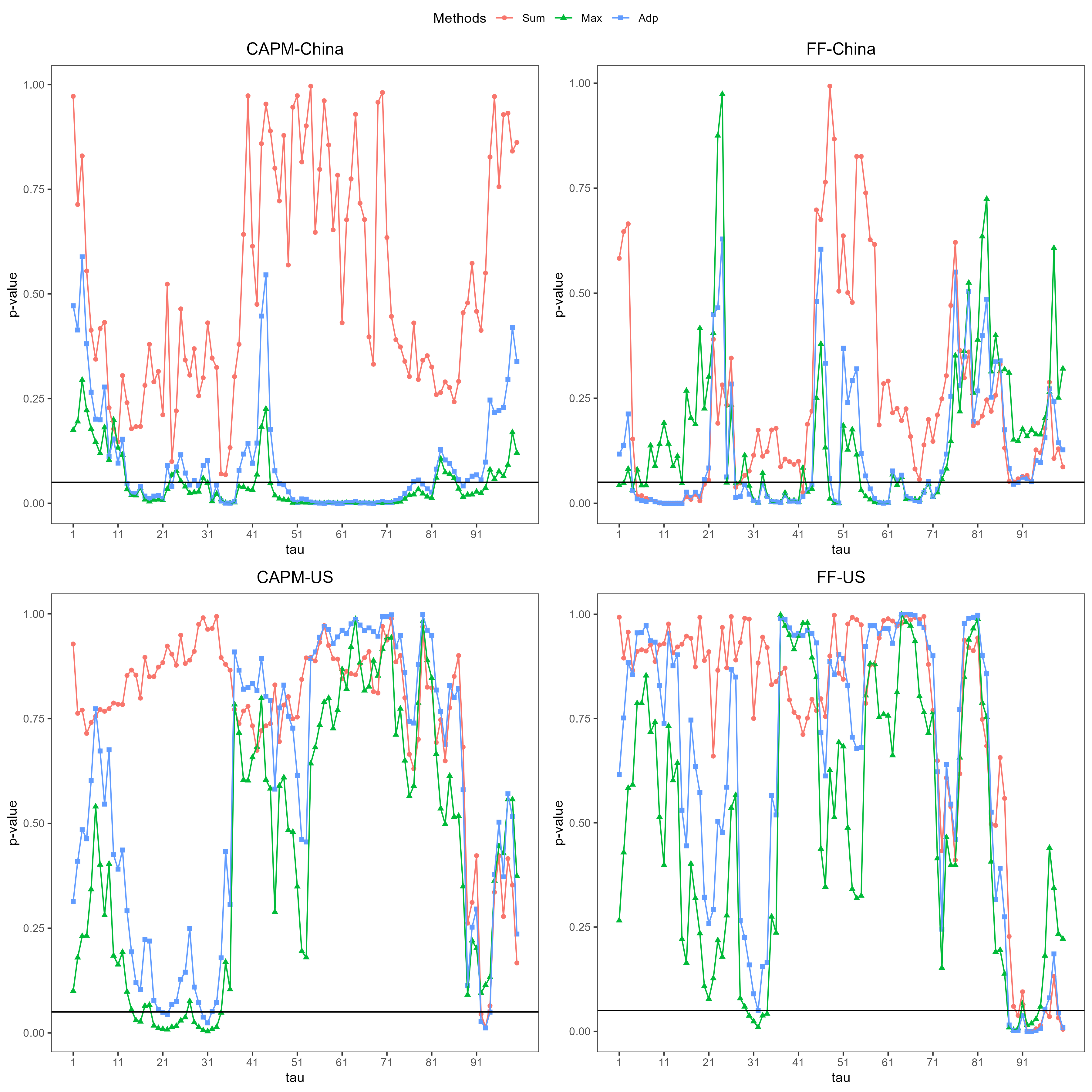}
    	\caption{P-values of alpha tests for Chinese and U.S.'s datasets, respectively.\label{fig:p-values}}
    \end{figure}

    \begin{figure}
    	\includegraphics[width=1\textwidth]{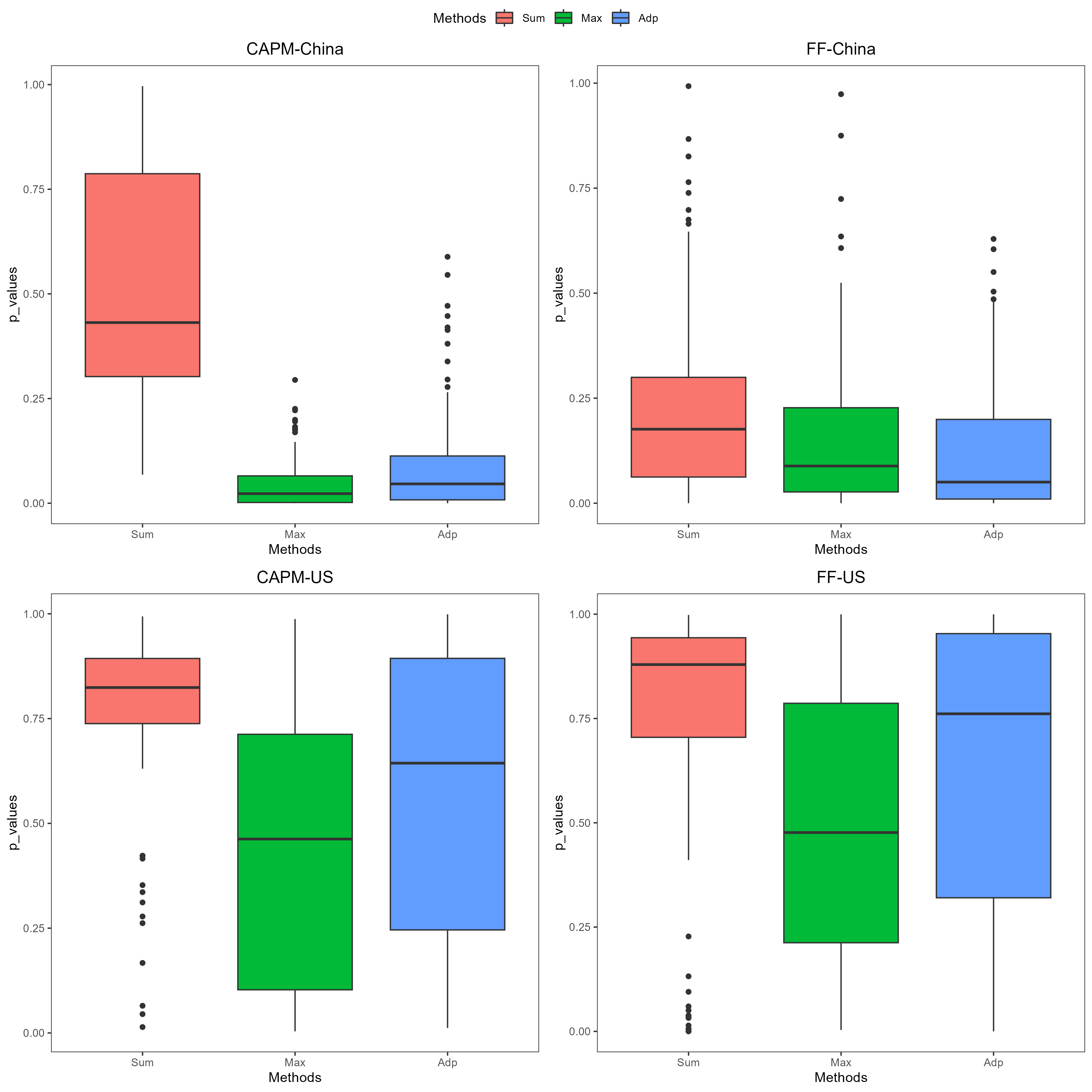}
    	\caption{Box-plot of p-values of alpha tests for Chinese and U.S.'s datasets, respectively. \label{fig:box-p}}
    \end{figure}

\section{Conclusion}
\label{sec:conc}
    In this paper, we propose the Max test for alpha testing under high dimensional case, which aims to simultaneously alleviate the difficulties of the time-variation in the risk-exposure coefficients and the sparse alternatives. Considering the good performance of Sum test under dense alternative, we suggest an adaptive test, which integrates the power advantages of both Max and Sum and is applicable for various alternatives. Moreover, the usefulness of the proposed Max and Adp tests is illustrated by two empirical examples.

    To further broaden the usefulness of our proposed tests, we conclude this article by identifying the following possible research avenues. First, if the goal is to identify the significance of alphas for all possible assets, then one can apply our Adp test in a multiple testing procedure to control the false discovery rate \citep{Giglio2020ThousandsOA}. Second, the assumption of sub-Gaussian-type tails may be too restrictive for stock data. \citet{zhao2023} proposed a robust high-dimensional alpha test based on spatial-sign for conditional time-varying factor models. How to construct an adaptive test for time-varying factor pricing models with heavy-tailed errors deserves some further studies. Finally, for sparse alternative, one may consider an L-statistic which combines the first several largest signals together, see more information in Remark 1 in \citet{chang2022}.

    \section{Appendix}
    \subsection{Proof of Theorem \ref{T2.1}}
    \begin{proof}
    	Define $\rho_{i0t}=\delta_i(t/T)-\bm{\lambda}_{0,i0}'\widetilde{\bm{B}}(t/T)$ and $\rho_{ijt}=\beta_{ij}(t/T)-\lambda_{0,ij}'B(t/T)$ for $1\leq j\leq d$ and $1\leq i\leq N$. Denote $\bm{\lambda}_i^0=(\bm{\lambda}_{0,ij}',0\le j\le d)'$, $\rho_{it}=\rho_{i0t}+\sum_{j=1}^d\rho_{ijt}f_{jt}$ and $\bm{\rho}_{i.}=(\rho_{i1},...,\rho_{iT})'$. Then
    	\begin{equation*}
    	\bm{R}_{i.}=\delta_i^0\bm{1}_T+\bm{Z}\bm{\lambda}_i^0+\bm{e}_{i.}+\bm{\rho}_{i.}, \text{ and } \widehat{\bm{e}}_{i.}=M_{\bm{Z}}\bm{e}_{i.}+M_{\bm{Z}}\bm{\rho}_{i.}+\delta_i^0M_{\bm{Z}}\bm{1}_T.
    	\end{equation*}
    	Accordingly, $t_i^2$ given in (\ref{equ.7}) can be written as
    	\begin{align*}\label{A.1}
    	t_i^2
    	&=T^{-1}\widehat{\sigma}_{ii}^{-1}(M_{\bm{Z}}\bm{e}_{i.}+M_{\bm{Z}}\bm{\rho}_{i.}+\delta_i^0M_{\bm{Z}}\bm{1}_T)' \bm{1}_T\bm{1}_T' (M_{\bm{Z}}\bm{e}_{i.}+M_{\bm{Z}}\bm{\rho}_{i.}+\delta_i^0M_{\bm{Z}}\bm{1}_T)\\
    	&=T^{-1}\widehat{\sigma}_{ii}^{-1}\{
    	\bm{e}_{i.}'M_{\bm{Z}}\bm{1}_T\bm{1}_T'M_{\bm{Z}}\bm{e}_{i.}
    	+\bm{\rho}_{i.}'M_{\bm{Z}}\bm{1}_T\bm{1}_T'M_{\bm{Z}}\bm{\rho}_{i.}
    	+2\bm{\rho}_{i.}'M_{\bm{Z}}\bm{1}_T\bm{1}_T'M_{\bm{Z}}\bm{e}_{i.}\\
    	&\qquad+2\delta_i^0\bm{1}_T'M_{\bm{Z}}\bm{1}_T\bm{1}_T'M_{\bm{Z}}\bm{e}_{i.}
    	+2\delta_i^0\bm{1}_T'M_{\bm{Z}}\bm{1}_T\bm{1}_T'M_{\bm{Z}}\bm{\rho}_{i.}+(\delta_i^0)^2\bm{1}_T'M_{\bm{Z}}\bm{1}_T\bm{1}_T'M_{\bm{Z}}\bm{1}_T\}\\
    	&=: \widehat{\sigma}_{ii}^{-1}(\varphi_i+\zeta_{i1}+\zeta_{i2}+\zeta_{i3}+\zeta_{i4}+\zeta_{i5})\tag{A.1}.
    	\end{align*}
    	That is, under $H_0$, $M_{NT}=\max_{1\le i\le N}\widehat{\sigma}_{ii}^{-1}(\varphi_i+\zeta_{i1}+\zeta_{i2})$. Thus, it suffices to show that
    	\begin{equation*}
    	\mathbb{P}\left( \max_{1\leq i\leq N}\varphi_i/\widehat{\sigma}_{ii}-2\log(N)+\log\{\log(N)\}\leq x\right) \to F(x), \text{ and }\max_{1\leq i\leq N}\widehat{\sigma}_{ii}^{-1}\zeta_{ik}=o_p(1),
    	\end{equation*}
    	for $k=1,2$.
    	\\~\\\textbf{Step 1.}
    	Show that
    	\begin{equation*}\label{A.2}
    	\mathbb{P}\left( \max_{1\leq i\leq N}\varphi_i/\widehat{\sigma}_{ii}-2\log(N)+\log\{\log(N)\}\leq x\right) \to F(x).\tag{A.2}
    	\end{equation*}
    	
    	Define $\bm{h}=(h_1,...,h_T)'=M_{\bm{Z}}\bm{1}_T$, $V_{it}=e_{it}h_t/\sigma_{ii}^{1/2}$ and $\widehat{V}_{it}=V_{it}I(V_{it}\leq \tau_{NT})$ for $i=1,...,N; t=1,...,T$, where $\tau_{NT}=2\zeta_{NT}^{-1}\eta^{-1/2}\sqrt{\log(T+N)}$ under Assumption \ref{A2.3}  (\romannumeral3). Here $\zeta_{NT}\to0$ will be specified later. Let $W_i=\sum_{t=1}^TV_{it}/\sqrt{T}$ and $\widehat{W}_i=\sum_{t=1}^T\widehat{V}_{it}/\sqrt{T}$. Notice that $W_i^2=\sigma_{ii}^{-1}\varphi_i$.
    	\\~\\\textbf{Step 1.1.}
    	Show that for any $x\in\mathbb{R}$,
    	\begin{equation*}\label{A.3}
    	\mathbb{P}\left(\max_{1\le i\le N}\widehat{W}_i^2-2\log(N)+\log\{\log(N)\}\le x \right)\to F(x).\tag{A.3}
    	\end{equation*}
    	By the Bonferroni inequality in Lemma \ref{lemma1}, for any fixed integer $k\leq N/2$,
    	\begin{align*}
    	\sum_{t=1}^{2k}(-1)^{t-1}\sum_{1\leq i_1<...<i_t\leq N}\mathbb{P}(|\widehat{W}_{i_1}|\geq x_N&,... |\widehat{W}_{i_t}|\geq x_N)
    	\leq\mathbb{P}(\max_{1\leq i\leq N}|\widehat{W}_i|\geq x_N)\\
    	\leq&\sum_{t=1}^{2k-1}(-1)^{t-1}\sum_{1\leq i_1<...<i_t\leq N}\mathbb{P}(|\widehat{W}_{i_1}|\geq x_N,... |\widehat{W}_{i_t}|\geq x_N),
    	\end{align*}
    	where $x_N=\sqrt{2\log{N}-\log\{\log(N)\}+x}$. Define $|\widehat{\bm{W}}|_{\min}=\min_{1\leq l\leq t}|\widehat{W}_{i_l}|$. Then, by Theorem 1 in \citet{Zaitsev1987OnTG}, we have
    	\begin{align*}
    	\mathbb{P}(|\widehat{\bm{W}}|_{\min}\geq x_N)\leq&
    	\mathbb{P}\left( |\bm{z}|_{\min}\geq x_N-k_n\log^{-1/2}(N)\right)\\
    	&\quad+c_1t^{5/2}\exp[-T^{1/2}k_n/\{c_2t^3\tau_{NT}\log^{1/2}(N)\}],
    	\end{align*}
    	where $c_1,c_2>0$ are two constants, $k_n\to 0$ (specified later), $|\bm{z}|_{\min}=\min_{1\leq l\leq t}|z_l|$, $\bm{z}=(z_1,...,z_t)'$ is a t-dimensional normal vector, which has the same covariance as $(\widehat{W}_{i_1},...,\widehat{W}_{i_t})'$.
    	
    	Because $\log(N)=o(T^{1/4})$, we let $k_n\to 0$ sufficiently slowly so that
    	\begin{equation*}
    	c_1t^{5/2}\exp[-T^{1/2}k_n/\{c_2t^3\tau_{NT}\log^{1/2}(N)\}]=O(N^{-\zeta})
    	\end{equation*}
    	for any large $\zeta>0$. Thus
    	\begin{equation*}
    	\mathbb{P}\left( \max_{1\leq i\leq N}|\widehat{W}_i|\geq x_N\right) \\
    	\leq\sum_{t=1}^{2k-1}(-1)^{t-1}\sum_{1\leq i_1<...<i_t\leq N}\mathbb{P}\left( |\bm{z}|_{\min}\geq x_N-k_n\log^{-1/2}(N)\right) +o(1).
    	\end{equation*}
    	Likewise,
    	\begin{equation*}
    	\mathbb{P}\left( \max_{1\leq i\leq N}|\widehat{W}_i|\geq x_N\right) \\
    	\ge\sum_{t=1}^{2k}(-1)^{t-1}\sum_{1\leq i_1<...<i_t\leq N}\mathbb{P}\left( |\bm{z}|_{\min}\geq x_N-k_n\log^{-1/2}(N)\right) +o(1).
    	\end{equation*}
    	
    	Define $V_t=\sum_{1\leq i_1<...<i_t\leq N}P_{i_1,...,i_t}=\sum_{1\leq i_1<...<i_t\leq N}\mathbb{P}(|z_{i_1}|\geq \widetilde{x}_N,...,|z_{i_t}|\geq \widetilde{x}_N)$, where $\widetilde{x}_N=x_N-k_n\log^{-1/2}(N)$. Write correlation matrix $\bm{R}=\bm{Q}'\bm{\Lambda}\bm{Q}$, where $\bm{Q}=(q_{ij})_{N\times N}$ is an orthogonal matrix and $\bm{\Lambda}=\diag(\lambda_1,...,\lambda_N)$, $\lambda_i$'s are the eigenvalues of $\bm{R}$. Since $\sum_{1\leq j\leq N}r_{ij}^2$ is the $i$th diagonal element of $\bm{R}^2=\bm{Q}^T\bm{\Lambda}^2\bm{Q}$, we have $\sum_{1\leq j\leq N}r_{ij}^2=\sum_{l=1}^Nq_{li}^2\lambda_l^2\leq c_3^2$ for some constant $c_3$ due to Assumption \ref{A2.4} (\romannumeral2). Thus, together with Assumption \ref{A2.4} (\romannumeral3), $\bm{R}$ satisfies the condition in Lemma \ref{lemma2}.
    	
    	Define $\mathcal{I}=\{1\leq i_1<...<i_t\leq N:\max_{1\leq k<l\leq t}|\cov(z_{i_k},z_{i_l})|\geq N^{-\gamma}\}$, where $\gamma>0$ is a sufficiently small number to be specified later. For $2\leq d\leq t-1$, define $\mathcal{I}_d=\{1\leq i_1<...<i_t\leq N:\card(S)=d, \text{ where } S \text{ is the largest subset of } \{i_1,...,i_t\} \text{ such that }\forall i_k,i_t\in S,i_k\ne i_t,|\cov(z_{i_k},z_{i_t})|<N^{-\gamma}\}$. Here, $\card(S)$ is the cardinality of $S$. For $d=1$, define $\mathcal{I}_1=\{1\leq i_1<...<i_t\leq N:|\cov(z_{i_k},z_{i_l})|\geq N^{-\gamma},\forall 1\leq k<l\leq t\}$. Hence, $\mathcal{I}=\cap_{d=1}^{t-1}\mathcal{I}_d$. For a fixed subset $S$ with $\card(S)=d$, the number of $i$ such that $|\cov(z_i,z_j)|\geq N^{-\gamma}$ for some $j\in S$ is no more than $c_4dN^{2\gamma}$, where $c_4$ is a constant. Indeed, if the number of such $i$ had been larger than $c_4dN^{2\gamma}$, then
    	\begin{equation*}
    	\sum\nolimits_{1\leq i\leq N}r_{ij}^2\geq \sum\nolimits_{\{i:|\cov(z_i,z_j)|\geq N^{-\gamma}\}}r_{ij}^2\geq N^{-2\gamma}c_4dN^{2\gamma}=c_4d>c_3^2,
    	\end{equation*}
    	for $c_4>d^{-1}c_3^2$, which is a contradiction. Note that the total number of $S$ with $\card(S)=d$ is $\binom{N}{d}$. This leads to
    	\begin{align*}
    	&\card(\mathcal{I}_d)\leq \binom{N}{d}\binom{c_4dN^{2\gamma}}{t-d}\leq N^d(c_4d)^{t-d}N^{2\gamma(t-d)}\leq (c_4t)^tN^{d+2\gamma t},\text{ and}\\
    	&\card(\mathcal{I})\leq \sum_{d=1}^{t-1}(c_4t)^tN^{d+2\gamma t}\leq (c_4t)^ttN^{2\gamma t+t-1}=O(N^{2\gamma t+t-1}).
    	\end{align*}
    	
    	Define $\mathcal{I}^c=\{1\leq i_1<...<i_t\leq N\}\backslash\mathcal{I}$. Then
    	\begin{equation*}
    	\card(\mathcal{I}^c)=\binom{N}{t}-O(N^{2\gamma t+t-1})=\{1+o(1)\}\binom{N}{t}
    	\end{equation*}
    	as long as $\gamma<(2t)^{-1}$. Thus, by equation (20) in the proof of Lemma 6 in \cite{TonyCai2014TwosampleTO}, we have
    	\begin{equation*}
    	P_{i_1,...,i_t}=\{1+o(1)\}\pi^{-t/2}N^{-t}\exp\left( -\frac{tx}{2}\right)
    	\end{equation*}
    	uniformly in $(i_1,...,i_t)\in\mathcal{I}^c$. Similarly, by equation (21) in the proof of Lemma 6 in \citet{TonyCai2014TwosampleTO}, we have for $1\leq d\leq t-1$
    	\begin{equation*}
    	\sum_{(i_1,...,i_t)\in\mathcal{I}_d}P_{i_1,...,i_t}\to 0.
    	\end{equation*}
    	Thus,
    	\begin{align*}
    	V_t
    	&=\sum_{(i_1,...,i_t)\in\mathcal{I}^c}P_{i_1,...,i_t}+\sum_{(i_1,...,i_t)\in\mathcal{I}}P_{i_1,...,i_t}\\
    	&=\sum_{(i_1,...,i_t)\in\mathcal{I}^c}P_{i_1,...,i_t}+\sum_{d=1}^{t-1}\sum_{(i_1,...,i_t)\in\mathcal{I}_d}P_{i_1,...,i_t}\\
    	&=\{1+o(1)\}\pi^{-t/2}\frac{1}{t!}\exp\left(-\frac{tx}{2} \right) .
    	\end{align*}
    	Combining all facts together, we get
    	\begin{align*}
    	\sum_{t=1}^{2k}(-1)^{t-1}\pi^{-t/2}\frac{1}{t!}\exp\left(-\frac{tx}{2} \right) \{1+o(1)\}&
    	\leq \mathbb{P}\left(\max_{1\leq i\leq N}|\widehat{W}_i|\geq x_N \right)\\
    	&\leq \sum_{t=1}^{2k-1}(-1)^{t-1}\pi^{-t/2}\frac{1}{t!}\exp\left(-\frac{tx}{2} \right) \{1+o(1)\}.
    	\end{align*}
    	Letting $k\to\infty$, we have (\ref{A.3}).
    	\\~\\\textbf{Step 1.2.}
    	Show that
    	\begin{equation*}\label{A.4}
    	\mathbb{P}\left( \max_{1\leq i\leq N}|W_i-\widehat{W}_i|\geq \frac{1}{\log(N)}\right) \to 0.\tag{A.4}
    	\end{equation*}
    	Notice that
    	\begin{equation*}
    	\mathbb{P}\left( \max_{1\leq i\leq N}|W_i-\widehat{W}_i|\geq \frac{1}{\log(N)}\right)
    	\leq \mathbb{P}\left( \max_{1\leq i\leq N}\max_{1\leq t\leq T}|V_{it}|\geq \tau_{NT}\right) \leq \sum_{1\leq i\leq N}\sum_{1\leq t\leq T}\mathbb{P}(|V_{it}|\geq \tau_{NT}).
    	\end{equation*}
    	For any $\zeta_t\to 0$, we have
    	\begin{align*}
    	\mathbb{P}(|V_{it}|\geq \tau_{NT})
    	&=\mathbb{P}(|e_{it}/\sigma_{ii}^{1/2}|\geq |h_t^{-1}|\tau_{NT})\\
    	&\leq \mathbb{P}(|e_{it}/\sigma_{ii}^{1/2}|\geq |h_t^{-1}|\tau_{NT}, |h_t^{-1}|\geq \zeta_{NT})+\mathbb{P}(|e_{it}/\sigma_{ii}^{1/2}|\geq |h_t^{-1}|\tau_{NT}, |h_t^{-1}|<\zeta_{NT})\\
    	&\leq \mathbb{P}(|e_{it}/\sigma_{ii}^{1/2}|\geq \tau_{NT}\zeta_{NT})+\mathbb{P}( |h_t|\geq\zeta_{NT}^{-1}).
    	\end{align*}
    	Under Assumption \ref{A2.3} (\romannumeral3), by the Markov inequality, we have
    	\begin{equation*}
    	\mathbb{P}(|e_{it}/\sigma_{ii}^{1/2}|\geq \tau_{NT}\zeta_{NT})\leq K\exp(-\eta\zeta_{NT}^2\tau_{NT}^2).
    	\end{equation*}
    	Next, we consider $\mathbb{P}( |h_t|\geq\zeta_{NT}^{-1})$. By Bernstein's inequality in \citet{Bosq1996NonparametricSF} and the same proof for Lemma A.8 of \citet{Ma2011SplinebackfittedKS}, under $L^3T^{-1}=o(1)$, we have
    	\begin{equation*}
    	||\bm{Z}'\bm{1}_T/T-\mathbb{E}(\bm{Z}'\bm{1}_T/T)||_\infty=O_{\text{a.s.}}\{\log(T)/\sqrt{TL}\}.
    	\end{equation*}
    	Further, under Assumption \ref{A2.2} (\romannumeral2), we have
    	\begin{align*}
    	||\mathbb{E}(\bm{Z}'\bm{1}_T/T)||_\infty
    	&\le\max_{1\le l\le L}T^{-1}  \sum_{t=1}^T|B_l(t/T) |\left(1+\sum_{j=1}^d\mathbb{E}|f_{jt}| \right)\\
    	&\le MT^{-1} \max_{1\le l\le L}\sum_{t\in\{t:|l(t)-l|\le q-1\}}|B_l(t/T)|\\
    	&\le ML^{-1}
    	\end{align*}
    	for some constant $0<M<\infty$, which leads to $||\bm{Z}'\bm{1}_T/T||_\infty=O_{\text{a.s.}}(L^{-1})$. By Lemma \ref{lemma4} and the result in \citet{Demko1986SpectralBF}, we have with probability one,
    	\begin{equation*}
    	||(\bm{Z}'\bm{Z}/T)^{-1}||_\infty\le c_5L,
    	\end{equation*}
    	for some constant $0<c_5<\infty$, as $T\to\infty$. Then, using the fact that
    	$\sum_{l=1}^LB_l(t/T)$ is bounded, we have
    	\begin{align*}\label{A.5}
    	|h_t|
    	&=|1-\bm{Z}_t'(\bm{Z}'\bm{Z})^{-1}\bm{Z}'\bm{1}_T|\nonumber\\
    	&\le 1+\sum_{k=1}^{(1+d)L}|Z_{tk}|\cdot ||(\bm{Z}'\bm{Z}/T)^{-1}||_\infty\cdot ||\bm{Z}'\bm{1}_T/T||_\infty\nonumber\\
    	&\le dc_6(1+||\bm{f}_t||),\tag{A.5}
    	\end{align*}
    	for some constant $0<c_6<\infty$, which leads to
    	\begin{align*}
    	\mathbb{P}(|h_t|\geq\zeta_{NT}^{-1})
    	\leq \mathbb{P}\left( ||\bm{f}_t||>d^{-1}c_6^{-1}\zeta_{NT}^{-1}-1\right)
    	\leq M(d^{-1}c_6^{-1}\zeta_{NT}^{-1}-1)^{-4(2+\kappa)},
    	\end{align*}
    	where the last inequality comes from Assumption \ref{A2.2} (\romannumeral2) and the Markov inequality. Combining all facts above and setting $\zeta_t=o((NT)^{-1/\{4(2+\kappa)\}})$, we obtain
    	\begin{align*}
    	\mathbb{P}\left( \max_{1\leq i\leq N}|W_i-\widehat{W}_i|\geq 1/\log(N)\right)
    	\leq&NT\left\lbrace M(d^{-1}c_6^{-1}\zeta_{NT}^{-1}-1)^{-4(2+\kappa)}
    	+K\exp(-\eta\zeta_{NT}^2\tau_{NT}^2) \right\rbrace \\
    	\leq&NT\left\lbrace M(d^{-1}c_6^{-1}\zeta_{NT}^{-1}-1)^{-4(2+\kappa)}
    	+K(N+T)^{-4} \right\rbrace\to 0.
    	\end{align*}
    	Hence, the proof of (\ref{A.4}) is complete.
    	\\~\\\textbf{Step 1.3.} Show that
    	\begin{equation*}\label{A.6}
    	\left|\max_{1\le i\le N}\varphi_i\widehat{\sigma}_{ii}^{-1}-\max_{1\le i\le N}W_i^2 \right|=o_p(1).\tag{A.6}
    	\end{equation*}
    	By (\ref{A.3}) and letting $x=1/2\log\{\log(N)\}$, we have
    	\begin{equation*}
    	\mathbb{P}\left(\max_{1\le i\le N}\widehat{W}_i^2\le 2\log(N)-\frac{1}{2}\log\{\log(N)\} \right)\to 1,
    	\end{equation*}
    	which, together with (\ref{A.4}), leads to
    	\begin{align*}
    	\left|\max_{1\leq i\leq N}W_i^2-\max_{1\leq i\leq N}\widehat{W}_i^2 \right|
    	\leq 2\max_{1\leq i\leq N}|W_i|\cdot \max_{1\leq i\leq N}|W_i-\widehat{W}_i|+\max_{1\leq i\leq N}|W_i-\widehat{W}_i|^2=o_p(1),
    	\end{align*}
    	and $\max_{1\leq i\leq N}W_i^2=O_p\{\log(N)\}$.
    	By Lemma E.2 and Proposition 3.1 in \citet{Fan2013PowerEI}, for some $c>0$,
    	\begin{equation*}\label{A.7}
    	\mathbb{P}\left(\max_{1\le i\le N}|\sigma_{ii}-\widehat{\sigma}_{ii}|\ge c\sqrt{\frac{\log(N)}{T}} \right)\to 0 \text{ and } \mathbb{P}\left(\frac{4}{9}\le \frac{\widehat{\sigma}_{ii}}{\sigma_{ii}}\le\frac{9}{4},i=1,...,N \right)\to 1.\tag{A.7}
    	\end{equation*}
    	Combining the results above, with probability tending to one, we have
    	\begin{align*}
    	\left|\max_{1\le i\le N}\varphi_i\widehat{\sigma}_{ii}^{-1}-\max_{1\le i\le N}W_i^2 \right|
    	\leq \max_{1\leq i\leq N}W_i^2\cdot \max_{1\leq i\leq N}\left|\frac{\sigma_{ii}}{\widehat{\sigma}_{ii}}-1 \right|
    	=O_p\{\log^{3/2}(N)T^{-1/2}\}\to 0,
    	\end{align*}
    	due to $\log(N)=o(T^{1/3})$. Then (\ref{A.6}) follows, which together with (\ref{A.3}) implies (\ref{A.2}).
    	\\~\\\textbf{Step 2.}
    	Show that
    	\begin{equation*}\label{A.8}
    	\max_{1\leq i\leq N}\widehat{\sigma}_{ii}^{-1}\zeta_{ik}=o_p(1), k=1,2.\tag{A.8}
    	\end{equation*}
    	
    	By Lemma \ref{lemma3} and Assumption \ref{A2.2}, for each $1\leq i\leq N$, we have
    	\begin{equation*}
    	\sup_{1\leq t\leq T}|\rho_{it}|=\sup_{1\leq t\leq T}\left| \rho_{i0t}+\sum_{j=1}^d\rho_{ijt}f_{jt}\right| =O(L^{-r}).
    	\end{equation*}
    	This, together with the fact that $\lambda_{\max}(M_{\bm{Z}}\bm{1}_T\bm{1}_T'M_{\bm{Z}})\leq T$, leads to
    	\begin{equation*}
    	\zeta_{i1}=T^{-1}(\bm{\rho}_{i.}'M_{\bm{Z}}\bm{1}_T)^2=O(TL^{-2r}),
    	\end{equation*}
    	for each $1\leq i\leq N$. Thus, by Assumption \ref{A2.4} (\romannumeral1) and (\ref{A.7}), we have $\max_{1\leq i\leq N}\widehat{\sigma}_{ii}^{-1}\zeta_{i1}=o_p(1)$.
    	
    	Define $\varpi_i=\sum_{s=1}^Th_s\rho_{is}$, then $\zeta_{i2}=2T^{-1}\sum_{t=1}^Te_{it}h_t\varpi_i$.
    	By (\ref{A.5}) and Assumption \ref{A2.2} (\romannumeral2), we obtain $\mathbb{E}(h_t^2)\leq \mathbb{E}\{dc_6(1+||\bm{f}_t||)\}^2\leq c_7$,
    	for some constant $0<c_7<\infty$. In addition, by Lemma \ref{lemma3}, we obtain $|\varpi_i|=O(TL^{-r})$. Furthermore, we have
    	\begin{equation*}
    	\text{Var}(\zeta_{i2})=4T^{-2}\sum_{t=1}^T\mathbb{E}(e_{it}h_t\varpi_i)^2=4T^{-2}\sum_{t=1}^T\sigma_{ii}\mathbb{E}(h_t\varpi_i)^2=O(TL^{-2r}).
    	\end{equation*}
    	Hence, $|\zeta_{i2}|=O_p(T^{1/2}L^{-r})$, which, together with Assumption \ref{A2.4} (\romannumeral1) and (\ref{A.7}), implies that $\max_{1\leq i\leq N}\widehat{\sigma}_{ii}^{-1}\zeta_{i2}=o_p(1)$.
    	
    	Then the proof of Theorem \ref{T2.1} is complete.
    \end{proof}

    \subsection{Proof of Proposition \ref{P2.1}}
    \begin{proof}
    	By (\ref{A.1}) and the triangle inequality, we have
    	\begin{align*}
    	M_{NT}
    	&=\max_{1\le i\le N} \widehat{\sigma}_{ii}^{-1} (\varphi_i+\zeta_{i1}+\zeta_{i2}+\zeta_{i3}+\zeta_{i4}+\zeta_{i5})\\
    	&\ge\max_{1\le i\le N} \widehat{\sigma}_{ii}^{-1} (\zeta_{i3}+\zeta_{i4}+\zeta_{i5})
    	-\max_{1\le i\le N}\widehat{\sigma}_{ii}^{-1}\varphi_i
    	-\max_{1\le i\le N}\widehat{\sigma}_{ii}^{-1}\zeta_{i1}
    	-\max_{1\le i\le N}\widehat{\sigma}_{ii}^{-1}\zeta_{i2}.
    	\end{align*}
    	According to the proof of Theorem \ref{T2.1}, we have $\max_{1\le i\le N}\widehat{\sigma}_{ii}^{-1}\zeta_{ik}=o_p(1)$ for $k=1,2$, and
    	\begin{equation*}
    	\mathbb{P}\left(\max_{1\le i\le N}\widehat{\sigma}_{ii}^{-1}\varphi_i\le 2\log(N)-\frac{1}{2}\log\{\log(N)\} \right)\to 1.
    	\end{equation*}
    	
    	Define $\varpi_i^*=\delta_i^0\bm{1}_T'M_{\bm{Z}}\bm{1}_T$, then $\zeta_{i3}=2T^{-1}\sum_{t=1}^Te_{it}h_t\varpi_i^*$ and $|\varpi_i^*|\leq T|\delta_i^0|$. Furthermore, due to (\ref{A.5}), we have
    	\begin{equation*}
    	\text{Var}(\zeta_{i3})=4T^{-2}\sum_{t=1}^T\mathbb{E}(e_{it}h_t\varpi_i^*)^2=4T^{-2}\sum_{t=1}^T\sigma_{ii}\mathbb{E}(h_t\varpi_i^*)^2=O( T|\delta_i^0|^2).
    	\end{equation*}
    	That is $\zeta_{i3}=O_p(T^{1/2}|\delta_i^0|)$. In addition, by Lemma \ref{lemma3}, we have
    	\begin{align*}
    	&\zeta_{i4}
    	=2T^{-1}\delta_i^0\bm{1}_T'M_{\bm{Z}}\bm{1}_T\bm{1}_T'M_{\bm{Z}}\bm{\rho}_{i.}\leq2T^{-1}|\delta_i^0|T\sum_{t=1}^T|h_t\rho_{it}|=O(|\delta_i^0|TL^{-r}), \text{ and}\\
    	&\zeta_{i5}
    	=T^{-1}(\delta_i^0)^2\bm{1}_T'M_{\bm{Z}}\bm{1}_T\bm{1}_T'M_{\bm{Z}}\bm{1}_T=T^{-1}(\delta_i^0)^2\left(\sum_{t=1}^Th_t \right)^2=O(T|\delta_i^0|^2).
    	\end{align*}
    	These results above, together with (\ref{A.7}), imply that the proposed $M_{NT}$-based test is consistent provided that $\max_{1\le i\le N}|\delta_i^0|\gtrsim\sqrt{\log(N)/T}$.
    \end{proof}

    \subsection{Proof of Theorem \ref{T3.1}}
    \begin{proof}
    	\textbf{Step 1.}
    	Investigate the asymptotic independence of $M_{NT}$ and $(S_{NT}-\widehat{\mu}_{NT})/\widehat{\sigma}_{NT}$ under Gaussian case, i.e. $\bm{e}_t\sim N(\bm{0},\bm{\Sigma})$.
    	Using the same notations given in (\ref{A.1}), we have
    	\begin{equation*}
    	S_{NT}=N^{-1}\sum_{i=1}^N(\varphi_i+\zeta_{i1}+\zeta_{i2}+\zeta_{i3}+\zeta_{i4}+\zeta_{i5}).
    	\end{equation*}
    	According to the proof of Theorem 1 and 2 in \citet{Ma2018TestingAI}, we have under $H_0$,
    	\begin{equation*}
    	(S_{NT}-\mu_{NT})/\sigma_{NT}=\varphi_{NT}/\sigma_{NT}+o_p(1),
    	\end{equation*}
    	where
    	\begin{equation*}
    	\varphi_{NT}=N^{-1}\sum_{i=1}^N\varphi_i-\mu_{NT}=2N^{-1}T^{-1}\sum_{t=2}^T\sum_{s=1}^{t-1}\bm{e}_t'\bm{e}_sh_th_s.
    	\end{equation*}
    	According to the proof of Theorem \ref{T2.1}, we have under $H_0$,
    	\begin{equation*}
    	M_{NT}=\max_{1\le i\le N}\sigma_{ii}^{-1}\varphi_i+o_p(1),
    	\end{equation*}
    	where $\varphi_i=T^{-1}(\sum_{t=1}^Te_{it}h_t)^2$. Hence, by Lemma \ref{lemma5}, it suffices to show that $\varphi_{NT}/\sigma_{NT}$ and $\max_{1\le i\le N}\varphi_i/\sigma_{ii}$ are asymptotically independent.
    	
    	For any fixed $x,y\in\mathbb{R}$, define $A_N=A_N(x)=\{\varphi_{NT}/\sigma_{NT}\le x\}$ and $B_i=B_i(y)=\{\varphi_i/\sigma_{ii}>2\log(N)-\log\{\log(N)\}+y\}$ for $i=1,...,N$. Then $\mathbb{P}(A_N)\to\Phi(x)$ and $\mathbb{P}(\cup_{i=1}^NB_i)\to1-F(y)$. Our goal is to prove that
    	\begin{equation*}
    	\mathbb{P}\left(\bigcup\limits_{i=1}^NA_NB_i \right)\to \Phi(x)\{1-F(y)\}.
    	\end{equation*}
    	
    	For each $d\ge 1$, define
    	\begin{align*}
    	&\zeta(N,d)=\sum_{1\le i_1<...<i_d\le N}|\mathbb{P}(A_NB_{i_1}...B_{i_d})-\mathbb{P}(A_N)\mathbb{P}(B_{i_1}...B_{i_d}) |, \\
    	&H(N,d)=\sum_{1\le i_1<...<i_d\le N}|\mathbb{P}(B_{i_1}...B_{i_d})|.
    	\end{align*}
    	By the inclusion-exclusion principle, we observe that for any integer $k\ge1$,
    	\begin{align*}
    	&\mathbb{P}\left(\bigcup\limits_{i=1}^NA_NB_i \right)\\
    	\le& \sum_{1\le i_1\le N}\mathbb{P}(A_NB_{i_1})-\sum_{1\le i_1<i_2\le N}\mathbb{P}(A_NB_{i_1}B_{i_2})+...+\sum_{1\le i_1<...<i_{2k}\le N}\mathbb{P}(A_NB_{i_1}...B_{i_{2k}})\\
    	\le& \mathbb{P}(A_N)\left\lbrace \sum_{1\le i_1\le N}\mathbb{P}(B_{i_1})-\sum_{1\le i_1<i_2\le N}\mathbb{P}(B_{i_1}B_{i_2})+...+\sum_{1\le i_1<...<i_{2k}\le N}\mathbb{P}(B_{i_1}...B_{i_{2k}})\right\rbrace  \\
    	&\qquad+\sum_{d=1}^{2k}\zeta(N,d)+H(N,2k+1)\\
    	\le& \mathbb{P}(A_N)\mathbb{P}\left(\bigcup\limits_{i=1}^NB_i \right)+\sum_{d=1}^{2k}\zeta(N,d)+H(N,2k+1).
    	\end{align*}
    	According to the proof of Theorem \ref{T2.1}, we have for each $d$,
    	\begin{equation*}
    	\lim\limits_{N\to\infty}H(N,d)=\pi^{-1/2}\frac{1}{d!}\exp\left(-\frac{dx}{2} \right).
    	\end{equation*}
    	We claim that for each $d$,
    	\begin{equation*}\label{C.1}
    	\lim\limits_{N\to\infty}\zeta(N,d)\to 0.\tag{C.1}
    	\end{equation*}
    	Then, by letting $k\to \infty$, we have
    	\begin{equation*}
    	\limsup\limits_{N\to\infty}\mathbb{P}\left(\bigcup\limits_{i=1}^NA_NB_i \right)\le \Phi(x)\{1-F(y)\}.
    	\end{equation*}
    	Likewise, we have
    	\begin{equation*}
    	\liminf\limits_{N\to\infty}\mathbb{P}\left(\bigcup\limits_{i=1}^NA_NB_i \right)\ge \Phi(x)\{1-F(y)\}.
    	\end{equation*}
    	Hence, the desired result follows.
    	
    	It remains to prove that the claim (\ref{C.1}) indeed holds.
    	
    	For each $t$, let $\bm{e}_{(1),t}=(e_{i_1,t},...,e_{i_d,t})'$, $\bm{e}_{(2),t}=(e_{i_{d+1},t},...,e_{i_N,t})'$, and for $k,l\in\{1,2\}$, let $\bm{\Sigma}_{kl}=\cov(\bm{e}_{(k),t},\bm{e}_{(l),t})$. By Lemma \ref{lemma6}, $\bm{e}_{(2),t}$ can be decomposed as $\bm{e}_{(2),t}=\bm{U}_t+\bm{V}_t$, where $\bm{U}_t=\bm{e}_{(2),t}-\bm{\Sigma}_{21}\bm{\Sigma}_{11}^{-1}e_{(1),t}$ and $\bm{V}_t=\bm{\Sigma}_{21}\bm{\Sigma}_{11}^{-1}\bm{e}_{(1),t}$ satisfying
    	\begin{equation*}
    	\bm{U}_t\sim N(\bm{0},\bm{\Sigma}_{22}-\bm{\Sigma}_{21}\bm{\Sigma}_{11}^{-1}\bm{\Sigma}_{12}),\quad\bm{V}_t\sim N(\bm{0},\bm{\Sigma}_{21}\bm{\Sigma}_{11}^{-1}\bm{\Sigma}_{12})\text{ and }
    	\bm{U}_t \text{ and } \bm{e}_{(1),t} \text{ are independent.}
    	\end{equation*}
    	Thus, we have
    	\begin{align*}
    	\varphi_{NT}
    	&=2N^{-1}T^{-1}\sum_{t=2}^T\sum_{s=1}^{t-1}\left(\bm{U}_t'\bm{U}_s+\bm{e}_{(1),t}'\bm{e}_{(1),s} +2\bm{V}_t'\bm{U}_s+\bm{V}_t'\bm{V}_s\right) h_th_s\\
    	&=:\varphi^*_1+\Theta_1+\Theta_2+\Theta_3,
    	\end{align*}
    	where $\varphi^*_1=2N^{-1}T^{-1}\sum_{t=2}^T\sum_{s=1}^{t-1}\bm{U}_t'\bm{U}_sh_th_s$, $\Theta_1=2N^{-1}T^{-1}\sum_{t=2}^T\sum_{s=1}^{t-1}\bm{e}_{(1),t}'\bm{e}_{(1),s} h_th_s$, $\Theta_2=2N^{-1}T^{-1}\sum_{t=2}^T\sum_{s=1}^{t-1}2\bm{V}_t'\bm{U}_s h_th_s$, $\Theta_2=2N^{-1}T^{-1}\sum_{t=2}^T\sum_{s=1}^{t-1}\bm{V}_t'\bm{V}_s h_th_s$ and $\varphi^*_2=\Theta_1+\Theta_2+\Theta_3$.
    	
    	We claim that for any $\epsilon>0$, $\exists$ a sequence of constants $c:=c_N>0$ with $c_N\to \infty$ s.t.
    	\begin{equation*}\label{C.2}
    	\mathbb{P}(|\Theta_k|\ge\epsilon\sigma_{NT})\le N^{-c}, k=1,2,3,\tag{C.2}
    	\end{equation*}
    	for sufficiently large $N$. Consequently, $\mathbb{P}(|\varphi^*_2/\sigma_{NT}|\ge\epsilon)\le N^{-c}$. Furthermore,
    	\begin{align*}
    	\mathbb{P}\left(A_N(x)B_{i_1}...B_{i_d}\right)
    	&\le\mathbb{P}\left(A_N(x)B_{i_1}...B_{i_d}, |\varphi^*_2/\sigma_{NT}|<\epsilon\right)+N^{-c}\\
    	&\le\mathbb{P}\left(|\varphi^*_1/\sigma_{NT}|<\epsilon+x, B_{i_1}...B_{i_d} \right)+N^{-c}\\
    	&=\mathbb{P}\left(|\varphi^*_1/\sigma_{NT}|<\epsilon+x \right)\mathbb{P}\left( B_{i_1}...B_{i_d} \right)+N^{-c}\\
    	&\le\left\lbrace \mathbb{P}\left(|\varphi^*_1/\sigma_{NT}|<\epsilon+x,|\varphi^*_2/\sigma_{NT}|<\epsilon  \right)+N^{-c}\right\rbrace \mathbb{P}\left( B_{i_1}...B_{i_d} \right)+N^{-c}\\
    	&\le \mathbb{P}\left(A_N(x+2\epsilon)\right)\mathbb{P}\left( B_{i_1}...B_{i_d} \right)+2N^{-c}.
    	\end{align*}
    	Likewise,
    	\begin{equation*}
    	\mathbb{P}\left(A_N(x)B_{i_1}...B_{i_d}\right)
    	\ge \mathbb{P}\left(A_N(x-2\epsilon)\right)\mathbb{P}\left( B_{i_1}...B_{i_d} \right)-2N^{-c}.
    	\end{equation*}
    	Hence,
    	\begin{equation*}
    	\left| \mathbb{P}\left(A_N(x)B_{i_1}...B_{i_d}\right)
    	-\mathbb{P}\left(A_N(x)\right)\mathbb{P}\left( B_{i_1}...B_{i_d} \right)\right|\le \Delta_{N,\epsilon}\cdot \mathbb{P}\left( B_{i_1}...B_{i_d}\right)+2N^{-c},
    	\end{equation*}
    	where
    	\begin{align*}
    	\Delta_{N,\epsilon}
    	&=\left|\mathbb{P}\left(A_N(x)\right)-\mathbb{P}\left(A_N(x+2\epsilon)\right) \right|+\left|\mathbb{P}\left(A_N(x)\right)-\mathbb{P}\left(A_N(x-2\epsilon)\right) \right|\\
    	&=\mathbb{P}\left(A_N(x+2\epsilon)\right)-\mathbb{P}\left(A_N(x-2\epsilon)\right)
    	\end{align*}
    	since $\mathbb{P}(A_N(x))$ is increasing in $x$. By running over all possible combinations of $1\le i_1<...<i_d\le N$, we have
    	\begin{equation*}
    	\zeta(N,d)\le \Delta_{N,\epsilon}\cdot H(N,d)+2\binom{N}{d}\cdot N^{-c}.
    	\end{equation*}
    	Since $\mathbb{P}(A_N(x))\to\Phi(x)$, we have $\lim_{\epsilon\downarrow 0}\limsup_{N\to\infty}\Delta_{N,\epsilon}=\lim_{\epsilon\downarrow 0}\{\Phi(x+2\epsilon)-\Phi(x-2\epsilon)\}=0$.
    	Since for each $d\ge 1$, $H(N,d)\to \pi^{-1/2}\exp(-dx/2)/d!$ as $N\to\infty$, we have $\limsup_{N\to\infty}H(N,d)<\infty$. Due to the fact that $\binom{N}{d}\le N^d$ for fixed $d\ge 1$, first sending $N\to\infty$ and then sending $\epsilon\downarrow0$, we get (\ref{C.1}).
    	
    	It remain to prove that the claim (\ref{C.2}) indeed holds.
    	
    	By (B.7) in \cite{Ma2018TestingAI}, $\exists$ constants $0<c_M<C_M<\infty$ s.t.
    	\begin{equation*}\label{C.3}
    	2c_M^2N^{-2}\tr(\bm{\Sigma}^2)\le\sigma_{NT}^2\le 2C_M^2N^{-2}\tr(\bm{\Sigma}^2)\{1+o(1)\}.\tag{C.3}
    	\end{equation*}
    	Define $\sigma_d^2=2c^2N^{-2}\tr(\bm{\Sigma}_{11}^2)$ with $c_M\le c\le C_M$, then
    	\begin{align*}
    	\mathbb{P}(|\Theta_1|\ge\epsilon\sigma_{NT})
    	&=\mathbb{P}\left( \left| 2N^{-1}T^{-1}\sum_{t=2}^T\sum_{s=1}^{t-1}\bm{e}_{(1),t}'\bm{e}_{(1),s} h_th_s\right| \ge\epsilon\sigma_{NT}\right)\\
    	&=\mathbb{P}\left( \left| \sigma_d^{-1}\cdot 2N^{-1}T^{-1} \sum_{t=2}^T\sum_{s=1}^{t-1}\bm{e}_{(1),t}'\bm{e}_{(1),s} h_th_s\right| \ge\epsilon'\sqrt{\frac{\tr(\bm{\Sigma}^2)}{\tr(\bm{\Sigma}_{11}^2)}}\right)\\
    	&\le\exp\left\lbrace -\epsilon' \sqrt{\frac{\tr(\bm{\Sigma}^2)}{\tr(\bm{\Sigma}_{11}^2)}} \right\rbrace  \cdot \mathbb{E}\left(\exp\left|\sigma_d^{-1}\cdot 2N^{-1}T^{-1} \sum_{t=2}^T\sum_{s=1}^{t-1}\bm{e}_{(1),t}'\bm{e}_{(1),s} h_th_s\right| \right)\\
    	&\le\exp\left\lbrace -\epsilon' \sqrt{\frac{\tr(\bm{\Sigma}^2)}{\tr(\bm{\Sigma}_{11}^2)}} \right\rbrace  \cdot \log(T),
    	\end{align*}
    	where the last inequality follows since
    	\begin{equation*}
    	\sigma_d^{-1}\cdot 2N^{-1}T^{-1} \sum_{t=2}^T\sum_{s=1}^{t-1}\bm{e}_{(1),t}'\bm{e}_{(1),s} h_th_s/\log\{\log(T)\}\to 0, \text{ a.s. }
    	\end{equation*}
    	by the law of the iterated logarithm of zero-mean square integrable martingale (see Theorem 4.8 in \citet{Hall1980MartingaleLT}). Similarly,
    	\begin{equation*}
    	\mathbb{P}(|\Theta_2|\ge\epsilon\sigma_{NT})
    	\le\exp\left\lbrace -\frac{\epsilon'}{2}\sqrt{\frac{\tr(\bm{\Sigma}^2)}{\tr(\bm{\Sigma}_{22\cdot 1}\bm{\Sigma}_{21}\bm{\Sigma}_{11}^{-1}\bm{\Sigma}_{12})}}\right\rbrace \cdot \log(T),
    	\end{equation*}
    	where $\bm{\Sigma}_{22\cdot 1}=\bm{\Sigma}_{22}-\bm{\Sigma}_{21}\bm{\Sigma}_{11}^{-1}\bm{\Sigma}_{12}$, and
    	\begin{equation*}
    	\mathbb{P}(|\Theta_3|\ge\epsilon\sigma_{NT})
    	\le\exp\left\lbrace -\frac{\epsilon'}{2}\sqrt{\frac{\tr(\bm{\Sigma}^2)}{\tr[(\bm{\Sigma}_{21}\bm{\Sigma}_{11}^{-1}\bm{\Sigma}_{12})^2]}}\right\rbrace  \cdot \log(T).
    	\end{equation*}
    	It is then easy to see that (\ref{C.2}) holds.
    	\\~\\\textbf{Step 2.}
    	Investigate the asymptotic independence of $M_{NT}$ and $(S_{NT}-\widehat{\mu}_{NT})/\widehat{\sigma}_{NT}$ when $\bm{e}_t$'s are sub-Gaussian.
    	
    	According to (\ref{C.3}), define
    	\begin{equation*}
    	W(\bm{e}_1,...,\bm{e}_T)=\sigma_s^{-1}\varphi_{NT}=\sigma_s^{-1}\left(2N^{-1}T^{-1}\sum_{t=2}^T\sum_{s=1}^{t-1}\bm{e}_t'\bm{e}_sh_th_s \right),
    	\end{equation*}
    	where $\sigma_s^2=2c^2N^{-2}\tr(\bm{\Sigma}^2)$ with $c_M\le c\le C_M$. For $\bm{X}=(x_1,...,x_q)'\in\mathbb{R}^q$, we consider a smooth approximation of the maximum function $\bm{X}\to \max_{1\le i\le q}x_i$, namely,
    	\begin{equation*}
    	F_{\beta}(\bm{X})=\beta^{-1}\log\left\lbrace\sum_{i=1}^q\exp(\beta x_i) \right\rbrace ,
    	\end{equation*}
    	where $\beta>0$ is the smoothing parameter that controls the level of approximation. An elementary calculation shows that $\forall \bm{X} \in\mathbb{R}^q$,
    	\begin{equation*}
    	0\le F_{\beta}(\bm{X})-\max_{1\le i\le q}x_i\le \beta^{-1}\log(q),
    	\end{equation*}
    	see \citet{Chernozhukov2013InferenceOC}. W.L.O.G. assume that $\sigma_{ii}=1$ for $i=1,...,N$. Define
    	\begin{equation*}
    	V(\bm{e}_1,...,\bm{e}_T)=F_{\beta}(\sqrt{\varphi_1},...,\sqrt{\varphi_N})=\beta^{-1}\log\left\lbrace\sum_{i=1}^N\exp\left(\beta T^{-1/2}\sum_{t=1}^Te_{it}h_t \right)  \right\rbrace .
    	\end{equation*}
    	By Lemma \ref{lemma5} and setting $\beta=T^{1/8}\log(N)$, it suffices to show that
    	\begin{equation*}
    	\mathbb{P}\left(W(\bm{e}_1,...,\bm{e}_T)\le x,V(\bm{e}_1,...,\bm{e}_T)\le \sqrt{2\log(N)-\log\{\log(N)\}+y} \right)\to\Phi(x)F(y).
    	\end{equation*}
    	Suppose $\{\bm{z}_1,...,\bm{z}_T\}$ are i.i.d. from $N(\bm{0},\bm{\Sigma})$, and are independent of $\{\bm{e}_1,...,\bm{e}_T\}$. According to the results of step 1, it remains to show that $(W(\bm{e}_1,...,\bm{e}_T), V(\bm{e}_1,...,\bm{e}_T))$ has the same limiting distribution as $(W(\bm{z}_1,...,\bm{z}_T),V(\bm{z}_1,...,\bm{z}_T))$.
    	
    	Let $\mathcal{C}_b^3(\mathbb{R})$ denote the class of bounded functions with bounded and continuous derivatives up to order 3. It is known that a sequence of random variables $\{Z_n\}_{n=1}^\infty$ converges weakly to a random variable $Z$ if and only if for every $f\in\mathcal{C}_b^3(\mathbb{R})$, $\mathbb{E}(f(Z_n))\to\mathbb{E}(f(Z))$, see, e.g. \citet{Pollard1984ConvergenceOS}. It suffices to show that
    	\begin{equation*}
    	\mathbb{E}\{f(W(\bm{e}_1,...,\bm{e}_T), V(\bm{e}_1,...,\bm{e}_T))\}-\mathbb{E}\{f(W(\bm{z}_1,...,\bm{z}_T),V(\bm{z}_1,...,\bm{z}_T))\}\to 0,
    	\end{equation*}
    	for every $f\in \mathcal{C}_b^3(\mathbb{R})$ as $(N,T)\to \infty$. We introduce $W_d=W(\bm{e}_1,...,\bm{e}_{d-1},\bm{z}_d,...,\bm{z}_T)$ and $V_d=V(\bm{e}_1,...,\bm{e}_{d-1},\bm{z}_d,...,\bm{z}_T)$ for $d=1,...,T+1$. Then
    	\begin{align*}
    	&|\mathbb{E}\{f(W(\bm{e}_1,...,\bm{e}_T), V(\bm{e}_1,...,\bm{e}_T))\}-\mathbb{E}\{f(W(\bm{z}_1,...,\bm{z}_T),V(\bm{z}_1,...,\bm{z}_T))\}|\\
    	\le&
    	\sum_{d=1}^T|\mathbb{E}\{f(W_d,V_d)\}-\mathbb{E}\{f(W_{d+1},V_{d+1})\}|.
    	\end{align*}
    	Let
    	\begin{align*}
    	&W_{d,0}=2N^{-1}T^{-1}\sigma_s^{-1}\left(\sum_{t=2}^{d-1}\sum_{s=1}^{t-1}\bm{e}_t'\bm{e}_sh_th_s+\sum_{t=d+2}^{T}\sum_{s=d+1}^{t-1}\bm{z}_t'\bm{z}_sh_th_s+\sum_{t=d+1}^{T}\sum_{s=1}^{d-1}\bm{z}_t'\bm{e}_sh_th_s \right),\\
    	&V_{d,0}=\beta^{-1}\log\left[\sum_{i=1}^N\exp\left\lbrace\beta T^{-1/2}\left(\sum_{t=1}^{d-1}e_{it}h_t +\sum_{t=d+1}^Tz_{it}h_t \right)  \right\rbrace  \right] ,
    	\end{align*}
    	which only rely on $\mathcal{F}_d=\sigma\{\bm{e}_1,...,\bm{e}_{d-1},\bm{z}_{d+1},...,\bm{z}_T\}$. By Taylor's expansion, we have
    	\begin{align*}
    	&f(W_d,V_d)-f(W_{d,0},V_{d,0})\\
    	=&f_1(W_{d,0},V_{d,0})(W_d-W_{d,0})+f_2(W_{d,0},V_{d,0})(V_d-V_{d,0})\\
    	&+\frac{1}{2}f_{11}(W_{d,0},V_{d,0})(W_d-W_{d,0})^2+\frac{1}{2}f_{22}(W_{d,0},V_{d,0})(V_d-V_{d,0})^2\\
    	&+\frac{1}{2}f_{12}(W_{d,0},V_{d,0})(W_d-W_{d,0})(V_d-V_{d,0})\\
    	&+O(|W_d-W_{d,0}|^3)+O(|V_d-V_{d,0}|^3),
    	\end{align*}
    	and
    	\begin{align*}
    	&f(W_{d+1},V_{d+1})-f(W_{d,0},V_{d,0})\\
    	=&f_1(W_{d,0},V_{d,0})(W_{d+1}-W_{d,0})+f_2(W_{d,0},V_{d,0})(V_{d+1}-V_{d,0})\\
    	&+\frac{1}{2}f_{11}(W_{d,0},V_{d,0})(W_{d+1}-W_{d,0})^2+\frac{1}{2}f_{22}(W_{d,0},V_{d,0})(V_{d+1}-V_{d,0})^2\\
    	&+\frac{1}{2}f_{12}(W_{d,0},V_{d,0})(W_{d+1}-W_{d,0})(V_{d+1}-V_{d,0})\\
    	&+O(|W_{d+1}-W_{d,0}|^3)+O(|V_{d+1}-V_{d,0}|^3),
    	\end{align*}
    	where $f=f(x,y)$, $f_1=\partial f/\partial x$, $f_2=\partial f/\partial y$, $f_{11}=\partial^2f/\partial x^2 $, $f_{22}=\partial^2f/\partial y^2 $ and $f_{12}=\partial^2f/\partial x\partial y $. Notice that
    	\begin{align*}\label{C.4}
    	&W_d-W_{d,0}=2N^{-1}T^{-1}\sigma_s^{-1}\left(\sum_{s=1}^{d-1}\bm{z}_d'\bm{e}_sh_dh_s+\sum_{t=d+1}^T\bm{z}_t'\bm{z}_dh_th_d \right),\text{ and}\\
    	&W_{d+1}-W_{d,0}=2N^{-1}T^{-1}\sigma_s^{-1}\left(\sum_{s=1}^{d-1}\bm{e}_d'\bm{e}_sh_dh_s+\sum_{t=d+1}^T\bm{z}_t'\bm{e}_dh_th_d \right). \tag{C.4}
    	\end{align*}
    	Due to $\mathbb{E}(\bm{e}_t)=\mathbb{E}(\bm{z}_t)=0$ and $\mathbb{E}(\bm{e}_t\bm{e}_t')=\mathbb{E}(\bm{z}_t\bm{z}_t')$, it can be verified that
    	\begin{align*}
    	&\mathbb{E}(W_d-W_{d,0}|\mathcal{F}_d)=\mathbb{E}(W_{d+1}-W_{d,0}|\mathcal{F}_d), \text{ and}\\
    	&\mathbb{E}\{(W_d-W_{d,0})^2|\mathcal{F}_d\}=\mathbb{E}\{(W_{d+1}-W_{d,0})^2|\mathcal{F}_d\}.
    	\end{align*}
    	Hence,
    	\begin{align*}
    	&\mathbb{E}\{f_1(W_d,V_{d,0})(W_{d+1}-W_{d,0})\}=\mathbb{E}\{f_1(W_{d,0},V_{d,0})(W_{d+1}-W_{d,0})\} \text{ and}\\
    	&\mathbb{E}\{f_{11}(W_{d,0},V_{d,0})(W_d-W_{d,0})^2\}=\mathbb{E}\{f_{11}(W_{d,0},V_{d,0})(W_{d+1}-W_{d,0})^2\}.
    	\end{align*}
    	
    	Next consider $V_d-V_{d,0}$. Let $v_{d,0,i}=T^{-1/2}\sum_{t=1}^{d-1}e_{it}h_t+T^{-1/2}\sum_{t=d+1}^Tz_{it}h_t$, $v_{d,i}=v_{d,0,i}+T^{-1/2}z_{id}h_d$, $v_{d+1,i}=v_{d,0,i}+T^{-1/2}e_{id}h_d$, $\bm{v}_{d,0}=(v_{d,0,1},...,v_{d,0,N})'$ and $\bm{v}_d=(v_{d,1},...,v_{d,N})'$. By Taylor's expansion, we have
    	\begin{align*}\label{C.5}
    	&V_d-V_{d,0}\\
    	=&\sum_{i=1}^N\partial_iF_\beta(\bm{v}_{d,0})(v_{d,i}-v_{d,0,i})+\frac{1}{2}\sum_{i=1}^N\sum_{j=1}^N\partial_i\partial_jF_\beta(\bm{v}_{d,0})(v_{d,i}-v_{d,0,i})(v_{d,j}-v_{d,0,j})\\
    	&+\frac{1}{6}\sum_{i=1}^N\sum_{j=1}^N\sum_{k=1}^N\partial_i\partial_j\partial_kF_\beta(\bm{v}_{d,0}+\delta(v_d-v_{d,0}))(v_{d,i}-v_{d,0,i})(v_{d,j}-v_{d,0,j})(v_{d,k}-v_{d,0,k}),\tag{C.5}
    	\end{align*}
    	for some $\delta\in(0,1)$. Again, due to $\mathbb{E}(\bm{e}_t)=\mathbb{E}(\bm{z}_t)=0$ and $\mathbb{E}(\bm{e}_t\bm{e}_t')=\mathbb{E}(\bm{z}_t\bm{z}_t')$, it can be verified that
    	\begin{equation*}
    	\mathbb{E}(v_{d,i}-v_{d,0,i}|\mathcal{F}_d)=\mathbb{E}(v_{d+1,i}-v_{d,0,i}|\mathcal{F}_d), \text{ and }
    	\mathbb{E}\{(v_{d,i}-v_{d,0,i})^2|\mathcal{F}_d\}=\mathbb{E}\{(v_{d+1,i}-v_{d,0,i})^2|\mathcal{F}_d\}.
    	\end{equation*}
    	By Lemma A.2 in \citet{Chernozhukov2013InferenceOC}, we have
    	\begin{equation*}
    	\left|\sum_{i=1}^N\sum_{j=1}^N\sum_{k=1}^N\partial_i\partial_j\partial_kF_\beta(\bm{v}_{d,0}+\delta(\bm{v}_d-\bm{v}_{d,0})) \right|\le C\beta^2
    	\end{equation*}
    	for some positive constant $C$. By Assumption \ref{A2.3} (\romannumeral2), we have $\mathbb{P}(\max_{1\le i\le N,1\le t\le T}|e_{it}|>C\log(NT))\to 0$, and since $z_{it}\sim N(0,1)$, $\mathbb{P}(\max_{1\le i\le N,1\le t\le T}|z_{it}|>C\log(NT))\to 0$. Hence,
    	\begin{align*}
    	&\left|\sum_{i=1}^N\sum_{j=1}^N\sum_{k=1}^N\partial_i\partial_j\partial_kF_\beta(\bm{v}_{d,0}+\delta(\bm{v}_d-\bm{v}_{d,0}))(v_{d,i}-v_{d,0,i})(v_{d,j}-v_{d,0,j})(v_{d,k}-v_{d,0,k}) \right|\\
    	\le& C\beta^2T^{-3/2}\log^3(NT)
    	\end{align*}
    	holds with probability approaching one. Consequently, we have with probability approaching one,
    	\begin{equation*}
    	\left|\mathbb{E}\{f_2(W_{d,0},V_{d,0})(V_d-V_{d,0})\}-\mathbb{E}\{f_2(W_{d,0},V_{d,0})(V_{d+1}-V_{d,0})\} \right|\le C\beta^2T^{-3/2}\log^3(NT).
    	\end{equation*}
    	Similarly, it can be verified that
    	\begin{equation*}
    	\left|\mathbb{E}\{f_{22}(W_{d,0},V_{d,0})(V_d-V_{d,0})^2\}-\mathbb{E}\{f_2(W_{d,0},V_{d,0})(V_{d+1}-V_{d,0})^2\} \right|\le C\beta^2T^{-3/2}\log^3(NT),
    	\end{equation*}
    	and
    	\begin{align*}
    	&\left|\mathbb{E}\{f_{12}(W_{d,0},V_{d,0})(W_d-W_{d,0})(V_d-V_{d,0})\}-\mathbb{E}\{f_{12}(W_{d,0},V_{d,0})(W_{d+1}-W_{d,0})(V_{d+1}-V_{d,0})\} \right|\\
    	\le& C\beta^2T^{-3/2}\log^3(NT).
    	\end{align*}
    	Again, Lemma A.2 in \citet{Chernozhukov2013InferenceOC}, together with  (\ref{C.5}), implies that  $\mathbb{E}(|V_d-V_{d,0}|^3)=O(T^{-3/2}\log^3(NT))$. According to (\ref{C.4}) and the proof of Lemma A.5 in \cite{Ma2018TestingAI}, we have $\mathbb{E}(|W_d-W_{d,0}|^4)=O(T^{-2})$, thus
    	\begin{equation*}
    	\sum_{d=1}^T\mathbb{E}(|W_d-W_{d,0}|^3)\le \sum_{d=1}^T\{\mathbb{E}(|W_d-W_{d,0}|^4)\}^{3/4}=O(T^{-1/2}).
    	\end{equation*}
    	Combining all facts together, we conclude that
    	\begin{equation*}
    	\sum_{d=1}^T|\mathbb{E}\{f(W_d,V_d)\}-\mathbb{E}\{f(W_{d+1},V_{d+1})\}|=O(\beta^2T^{-3/2}\log^3(NT))+O(T^{-1/2})\to 0,
    	\end{equation*}
    	as $(N,T)\to \infty$, due to $\log(N)=o(T^{1/4})$. Then the proof of Theorem \ref{T3.2} is complete.
    \end{proof}

    \subsection{Proof of Theorem \ref{T3.3}}
    \begin{proof}
    	It suffices to show that the conclusion holds for Gaussian $\bm{e}_t$'s. Using the notation given in (\ref{A.1}), we have
    	\begin{align*}
    	&M_{NT}=\max_{1\le i\le N}\widehat{\sigma}_{ii}^{-1}(\varphi_i+\zeta_{i1}+\zeta_{i2}+\zeta_{i3}+\zeta_{i4}+\zeta_{i5}) \text{ and }\\
    	&S_{NT}=N^{-1}\sum_{i=1}^N(\varphi_i+\zeta_{i1}+\zeta_{i2}+\zeta_{i3}+\zeta_{i4}+\zeta_{i5}).
    	\end{align*}
    	Under the alternative hypothesis given in (\ref{equ.11}), according to the proof of Theorem \ref{T2.1}, we have
    	\begin{align*}
    	M_{NT}
    	&=\max_{1\le i\le N}\sigma_{ii}^{-1}\left\lbrace \varphi_i+T^{-1}(\bm{1}_T'M_{\bm{Z}}\bm{1}_T)^2(\delta_i^0)^2 \right\rbrace +o_p(1)\\
    	&=\max_{i\in \mathcal{A}}\sigma_{ii}^{-1}\left\lbrace \varphi_i+T^{-1}(\bm{1}_T'M_{\bm{Z}}\bm{1}_T)^2(\delta_i^0)^2 \right\rbrace
    	+\max_{i\in \mathcal{A}^c}\sigma_{ii}^{-1} \varphi_i+o_p(1),
    	\end{align*}
    	where $\varphi_i=T^{-1}(\sum_{t=1}^Te_{it}h_t)^2$. In addition, by Lemma A.7 in \citet{Ma2018TestingAI}, we have
    	\begin{equation*}
    	(S_{NT}-\mu_{NT})/\sigma_{NT}=\sigma_{NT}^{-1}\varphi_{NT}+\sigma_{NT}^{-1}N^{-1}T^{-1}(\bm{1}_T'M_{\bm{Z}}\bm{1}_T)^2\sum_{i=1}^N(\delta_i^0)^2 +o_p(1),
    	\end{equation*}
    	where $\varphi_{NT}=2N^{-1}T^{-1}\sum_{t=2}^T\sum_{s=1}^{t-1}\bm{e}_t'\bm{e}_sh_th_s$. Define $\bm{e}_{\mathcal{A},t}=(e_{it},i\in\mathcal{A})'$ and $\bm{e}_{\mathcal{A}^c,t}=(e_{it},i\in\mathcal{A}^c)'$. Then, we rewrite
    	\begin{align*}
    	\varphi_{NT}
    	&=2N^{-1}T^{-1}\sum_{t=2}^T\sum_{s=1}^{t-1}\bm{e}_{\mathcal{A},t}'\bm{e}_{\mathcal{A},s}h_th_s
    	+2N^{-1}T^{-1}\sum_{t=2}^T\sum_{s=1}^{t-1}\bm{e}_{\mathcal{A}^c,t}'\bm{e}_{\mathcal{A}^c,s}h_th_s\\
    	&=:\varphi_{NT,\mathcal{A}}+\varphi_{NT,\mathcal{A}^c}.
    	\end{align*}
    	According to the proof of Theorem \ref{T3.2}, we have known that $\sigma_{NT}^{-1}\varphi_{NT,\mathcal{A}^c}$ and $\max_{i\in\mathcal{A}^c}\sigma_{ii}^{-1}\varphi_i$ are asymptotically independent. Hence, it suffices to show that $\sigma_{NT}^{-1}\varphi_{NT,\mathcal{A}^c}$ is asymptotically independent of $\varphi_i,i\in\mathcal{A}$.
    	
    	Define $\bm{\Sigma}_{\mathcal{A},\mathcal{A}^c}=\cov(\bm{e}_{\mathcal{A},t}, \bm{e}_{\mathcal{A}^c,t})$. By Lemma \ref{lemma6}, $\bm{e}_{\mathcal{A}^c,t}$ can be decomposed as $\bm{e}_{\mathcal{A}^c,t}=\bm{U}_t+\bm{V}_t$, where $\bm{U}_t=\bm{e}_{\mathcal{A}^c,t}-\bm{\Sigma}_{\mathcal{A}^c,\mathcal{A}}\bm{\Sigma}_{\mathcal{A},\mathcal{A}}^{-1}\bm{e}_{\mathcal{A},t}$ and $\bm{V}_t=\bm{\Sigma}_{\mathcal{A}^c,\mathcal{A}}\bm{\Sigma}_{\mathcal{A},\mathcal{A}}^{-1}\bm{e}_{\mathcal{A},t}$ satisfying that $\bm{U}_t\sim N(\bm{0},\bm{\Sigma}_{\mathcal{A}^c,\mathcal{A}^c}-\bm{\Sigma}_{\mathcal{A}^c,\mathcal{A}}\bm{\Sigma}_{\mathcal{A},\mathcal{A}}^{-1}\bm{\Sigma}_{\mathcal{A},\mathcal{A}^c})$, $\bm{V}_t\sim N(\bm{0},\bm{\Sigma}_{\mathcal{A}^c,\mathcal{A}}\bm{\Sigma}_{\mathcal{A},\mathcal{A}}^{-1}\bm{\Sigma}_{\mathcal{A},\mathcal{A}^c})$ and
    	\begin{equation*}\label{C.6}
    	\bm{U}_t \text{ and } \bm{e}_{\mathcal{A},t} \text{ are independent}.\tag{C.6}
    	\end{equation*}
    	Then, we have
    	\begin{equation*}
    	\varphi_{NT,\mathcal{A}^c}=2N^{-1}T^{-1}\sum_{t=2}^T\sum_{s=1}^{t-1} (\bm{U}_t'\bm{U}_s+2\bm{U}_t'\bm{V}_s +\bm{V}_t'\bm{V}_s) h_th_s.
    	\end{equation*}
    	By using the arguments similar to those in the proof of (\ref{C.2}), we have
    	\begin{align*}
    	&\mathbb{P}\left(2N^{-1}T^{-1}\sum_{t=2}^T\sum_{s=1}^{t-1} 2\bm{U}_t'\bm{V}_sh_th_s\ge \epsilon \sigma_{NT} \right)
    	\le \log(T)\exp(-c_{\epsilon}N^{1/2}|\mathcal{A}|^{-1/2})\to 0 \text{ and}\\
    	&\mathbb{P}\left(2N^{-1}T^{-1}\sum_{t=2}^T\sum_{s=1}^{t-1} \bm{V}_t'\bm{V}_sh_th_s\ge \epsilon \sigma_{NT} \right)
    	\le \log(T)\exp(-c_{\epsilon}N^{1/2}|\mathcal{A}|^{-1/2})\to 0,
    	\end{align*}
    	due to $|\mathcal{A}|=o (N/[\log\{\log(N)\}]^2)$ and $\log(N)=o(T^{1/4})$. Consequently, we conclude that
    	\begin{equation*}
    	\sigma_{NT}^{-1}\varphi_{NT,\mathcal{A}^c}=\sigma_{NT}^{-1}\cdot 2N^{-1}T^{-1}\sum_{t=2}^T\sum_{s=1}^{t-1} \bm{U}_t'\bm{U}_sh_th_s+o_p(1),
    	\end{equation*}
    	which, together with Lemma \ref{lemma5} and (\ref{C.6}), implies that $\sigma_{NT}^{-1}\varphi_{NT,\mathcal{A}^c}$ is asymptotically independent of $e_{it},i\in\mathcal{A}$. Hence, Theorem \ref{T3.3} follows.
    \end{proof}

    \subsection{Some useful facts}
    \begin{lemma}
    	\label{lemma1}
    	(Bonferroni inequality) Let $A=\cup_{t=1}^nA_t$. For any integer $k$, where $1\leq k\leq n/2$, we have
    	\begin{equation*}
    	\sum\nolimits_{t=1}^{2k}(-1)^{t-1}E_t\leq \mathbb{P}(A)\leq\sum\nolimits_{t=1}^{2k-1}(-1)^{t-1}E_t,
    	\end{equation*}
    	where $E_t=\sum_{1\leq i_1<...< i_t\leq n}\mathbb{P}(A_{i_1}\cap...\cap A_{i_t})$.
    \end{lemma}
    \begin{proof}
    	See Lemma 1 in \citet{TonyCai2014TwosampleTO}.
    \end{proof}

    \begin{lemma}
    	\label{lemma2}
    	Let $(Z_1,...,Z_N)'$ be a zero mean multivariate normal random vector with covariance matrix $\bm{\Sigma}=(\sigma_{ij})_{N\times N}$ and diagonal $\sigma_{ii}=1$ for $1\leq i\leq N$. Suppose that $\max\nolimits_{1\leq i<j\leq N}|\sigma_{ij}|\leq r<1$ and $\max\nolimits_{1\leq j\leq N}\sum\nolimits_{i=1}^N\sigma_{ij}^2\leq c$ for some $r$ and $c$. Then foe any $x\in\mathbb{R}$ as $N\to \infty$,
    	\begin{equation*}
    	\mathbb{P}\left[\max\nolimits_{1\leq i\leq N}Z_i^2-2\log(N)+\log\{\log(N)\}\leq x\right]\to\exp\left\lbrace-\pi^{-1/2}\exp(-x/2) \right\rbrace.
    	\end{equation*}
    \end{lemma}
    \begin{proof}
    	See Lemma 6 in \citet{TonyCai2014TwosampleTO}.
    \end{proof}
    \begin{lemma}
    	\label{lemma3}
    	Define $\rho_{i0t}=\delta_i(t/T)-\bm{\lambda}_{0,i0}'\widetilde{\bm{B}}(t/T)$ and $\rho_{ijt}=\beta_{ij}(t/T)-\bm{\lambda}_{0,ij}'\bm{B}(t/T)$ for $1\leq j\leq d$ and $1\leq i\leq N$. Then, under Assumption \ref{A2.1}, there exist $\bm{\lambda}_{0,i0}\in \mathbb{R}^L$ and $\bm{\lambda}_{0,ij}\in \mathbb{R}^L$ such that
    	\begin{equation*}
    	\sup\nolimits_{1\leq t\leq T}|\rho_{i0t}|=O(L^{-r}) \text{ and } \sup\nolimits_{1\leq t\leq T}|\rho_{ijt}|=O(L^{-r}) \text{ as } T\to\infty.
    	\end{equation*}
    \end{lemma}
    \begin{proof}
    	See Lemma A.1 in \citet{Ma2018TestingAI}.
    \end{proof}

    \begin{lemma}
    	\label{lemma4}
    	Under Assumption \ref{A2.2}, $\exists$ constants $0<c_1\leq C_1<\infty$ and $0>C_2<\infty$, with probability 1,
    	\begin{equation*}
    	c_1L^{-1}\leq \lambda_{\min}(\bm{Z}'\bm{Z}/T)\leq \lambda_{\max}(\bm{Z}'\bm{Z}/T) \leq C_1L^{-1},
    	\end{equation*}
    	\begin{equation*}
    	C_1L\leq \lambda_{\min}\{(\bm{Z}'\bm{Z}/T)^{-1}\}\leq \lambda_{\max}\{(\bm{Z}'\bm{Z}/T)^{-1}\} \leq c_1L,
    	\end{equation*}
    	as $T\to \infty$, and for any onozero vector $\bm{a}\in \mathbb{R}^T$ with $||\bm{a}||=1$, $\bm{a}'(\bm{Z}'\bm{Z}/T)\bm{a}\leq C_2L^{-1}$.
    \end{lemma}
    \begin{proof}
    	See Lemma A.2 in \citet{Ma2018TestingAI}.
    \end{proof}

    \begin{lemma}
    	\label{lemma5}
    	Let $\{(U,U_N,\widetilde{U}_N)\in\mathbb{R}^3;N\geq 1\}$ and $\{(V,V_N,\widetilde{V}_N)\in\mathbb{R}^3;N\geq 1\}$ be two sequences of random variables with $U_N\stackrel{d}{\to} U$ and $V_N\stackrel{d}{\to} V$ as $N\to\infty$. Assume $U$ and $V$ are continuous random variables. We assume that
    	\begin{equation*}
    	\widetilde{U}_N=U_N+o_p(1) \text{ and } \widetilde{V}_N=V_N+o_p(1).
    	\end{equation*}
    	If $U_N$ and $V_N$ are asymptotically independent, then $\widetilde{U}_N$ and $\widetilde{V}_N$ are also asymptotically independent.
    \end{lemma}
    \begin{proof}
    	See Lemma 7.10 in \citet{Feng2022AsymptoticIO}.
    \end{proof}

    \begin{lemma}
    	\label{lemma6}
    	Let $\bm{X}\sim N(\bm{\mu},\bm{\Sigma})$ with  invertible $\bm{\Sigma}$, and partition $\bm{X}$, $\bm{\mu}$ and $\bm{\Sigma}$ as
    	\begin{equation*}
    	\bm{X}=\left(\begin{array}{l}
    	\bm{X}_1\\\bm{X}_2
    	\end{array} \right), \quad
    	\bm{\mu}=\left(\begin{array}{l}
    	\bm{\mu}_1\\\bm{\mu}_2
    	\end{array} \right),\quad
    	\bm{\Sigma}=\left(\begin{array}{ll}
    	\bm{\Sigma}_{11}&\bm{\Sigma}_{12}\\
    	\bm{\Sigma}_{21}&\bm{\Sigma}_{22}
    	\end{array} \right).
    	\end{equation*}
    	Then $\bm{X}_2-\bm{\Sigma}_{21}\bm{\Sigma}_{11}^{-1}\bm{X}_1\sim N(\bm{\mu}_2-\bm{\Sigma}_{21}\bm{\Sigma}_{11}^{-1}\bm{\mu}_1,\bm{\Sigma}_{22\cdot 1})$ and is independent of $\bm{X}_1$, where $\bm{\Sigma}_{22\cdot 1}=\bm{\Sigma}_{22}-\bm{\Sigma}_{21}\bm{\Sigma}_{11}^{-1}\bm{\Sigma}_{12}$.
    \end{lemma}
    \begin{proof}
    	See Theorem 1.2.11 in \citet{Muirhead1982AspectsOM}.
    \end{proof}

\end{document}